\documentclass{article}

\usepackage{graphicx}
\usepackage{subcaption}
\usepackage{booktabs} % for professional tables
\usepackage{multicol}
\usepackage{multirow}
\usepackage{bm}
\usepackage{hyperref}

\usepackage[preprint]{icml2026}

\usepackage{amsmath}
\usepackage{amssymb}
\usepackage{mathtools}
\usepackage{amsthm}
\usepackage{algorithm}
\usepackage{algorithmic}

\usepackage[capitalize,noabbrev]{cleveref}

\theoremstyle{plain}
\newtheorem{theorem}{Theorem}[section]
\newtheorem{proposition}[theorem]{Proposition}
\newtheorem{lemma}[theorem]{Lemma}

\theoremstyle{definition}

\theoremstyle{remark}

\usepackage{graphicx}
\usepackage{booktabs}
\usepackage{amssymb}
\usepackage{multirow}
\usepackage{makecell}
\usepackage{pifont}
\usepackage[svgnames]{xcolor}
\usepackage{enumitem}

\usepackage[textsize=tiny]{todonotes}

\begin{document}

\twocolumn[
  \icmltitle{Objective Shaping with Hard Negatives: Windowed Partial AUC Optimization for
  RL-based LLM Recommenders}




  \begin{icmlauthorlist}
    \icmlauthor{Wentao Shi\textsuperscript{\textdagger}}{ustc}
    \icmlauthor{Qifan Wang}{meta}
    \icmlauthor{Chen Chen}{meta}
    \icmlauthor{Fei Liu}{meta}
    \icmlauthor{Dongfang Liu}{rit}
    \icmlauthor{Xu Liu}{meta}
    \icmlauthor{Wanli Ma}{meta}
    \icmlauthor{Junfeng Pan}{meta}
    \icmlauthor{Linhong Zhu}{meta}
    \icmlauthor{Fuli Feng}{ustc}
  \end{icmlauthorlist}

  \icmlaffiliation{ustc}{University of Science and Technology of China (USTC), Hefei, China}
  \icmlaffiliation{meta}{Meta AI, Menlo Park, USA}
  \icmlaffiliation{rit}{Rochester Institute of Technology, Rochester, USA}

  \icmlcorrespondingauthor{Qifan Wang}{wqfcr@fb.com}

  \icmlcorrespondingauthor{Fuli Feng}{fulifeng93@gmail.com}

  \icmlkeywords{Negative Sampling, Partial AUC, LLM-based Recommenders}

  \vskip 0.3in
]



% Add an author note for the \textsuperscript{\textdagger} symbol (first author)
\printAffiliationsAndNotice{\textsuperscript{\textdagger}Work done during an internship at Meta.}

\begin{abstract}
  Reinforcement learning (RL) effectively optimizes Large Language Model (LLM)-based recommenders by contrasting positive and negative items. Empirically, training with beam-search negatives consistently outperforms random negatives, yet the mechanism is not well understood. We address this gap by analyzing the induced optimization objective and show that: (i) Under binary reward feedback, optimizing LLM recommenders with Group Relative Policy Optimization (GRPO) is theoretically equivalent to maximizing the Area Under the ROC Curve (AUC), which is often misaligned with Top-$K$ recommendation; and (ii) Replacing random negatives with beam-search negatives reshapes the objective toward partial AUC, improving alignment with Top-$K$ metrics. Motivated by this perspective, we introduce \textbf{W}indowed \textbf{P}artial \textbf{AUC} (\textbf{WPAUC}), which constrains the false positive rate (FPR) to a window [$\alpha,\alpha+d$] to more directly align with Top-$K$ metrics. We further propose an efficient \textbf{T}hreshold-\textbf{A}djusted \textbf{Win}dowed reweighting (\textbf{TAWin}) RL method for its optimization, enabling explicit control over the targeted Top-$K$ performance. Experiments on four real-world datasets validate the theory and deliver consistent state-of-the-art performance.
\end{abstract}

\section{Introduction}

% windowed partial AUC

% first paragraph
% what is RS, LLM in RS，
% windowed partial AUC

% cut preliminary 
% ReRe 多解释一下，OPAUC WPAUC
% table 2 3 优化空间
% Qwen-2.5 background 
% 实验部分多解释一些

Recommendation systems are pivotal in mitigating information overload on large-scale platforms by delivering personalized Top-$K$ recommendations~\citep{rendle2012bpr, rendle2014improving}. With recent breakthroughs in LLMs~\citep{liu2024deepseek, openai2024gpt4technicalreport}, LLM-based recommenders have emerged as a promising paradigm for Top-$K$ recommendation~\citep{DBLP:conf/kdd/ChenSSH17, DBLP:conf/recsys/BaoZZWF023}. In this setting, RL serves as an effective optimization framework, refining user relative preference modeling by contrasting sampled positive items against negative items~\citep{DBLP:journals/corr/abs-2510-12211}. Empirical evidence, including our own reproductions (Figure~\ref{fig:AUC_comparison}(a)), confirms that RL training using negative items derived from beam search consistently outperforms random sampling~\citep{zhou2025onerec, kong2025minionerec}.

\begin{figure}[t]
    \centering
    \begin{subfigure}[t]{0.49\linewidth}
        \centering
        \includegraphics[width=\linewidth]{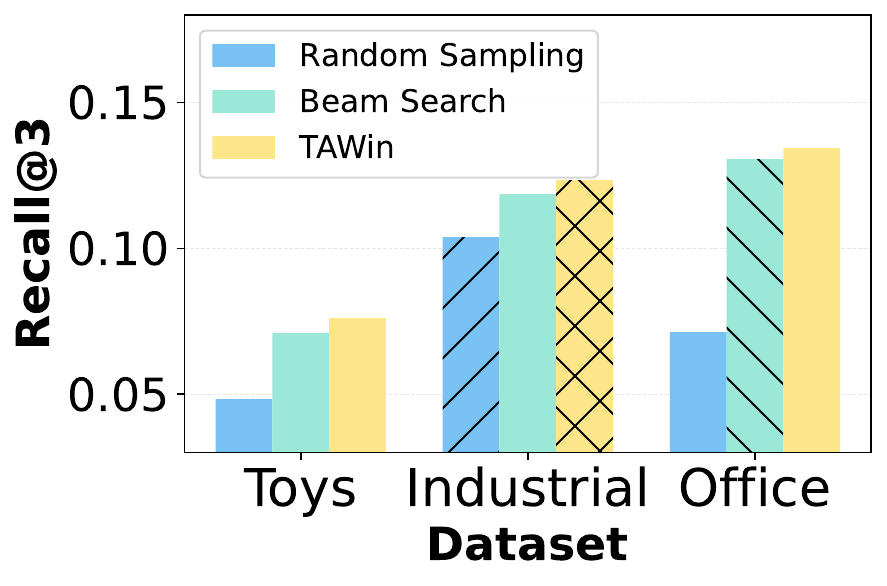}
        \caption{Performance Comparison}
        %\label{fig:sampling_comparison}
    \end{subfigure}
    \hfill
    \begin{subfigure}[t]{0.49\linewidth}
        \centering
        \includegraphics[width=\linewidth]{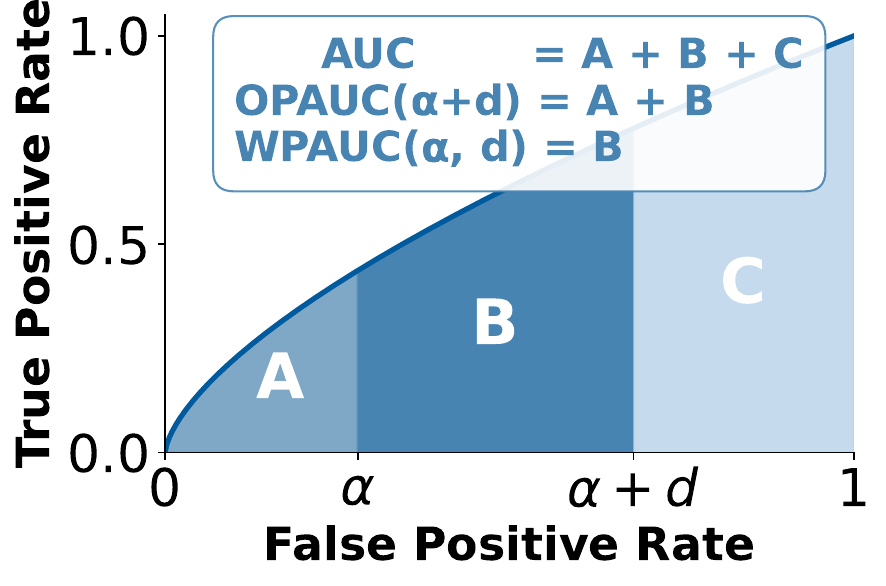}
        \caption{Partial AUC}
        %\label{fig:WPAUC}
    \end{subfigure}
    \caption{(a) Performance comparison of Random Sampling, Beam Search (\citep{DBLP:journals/corr/abs-2510-12211}), and our approach (\textbf{TAWin}) on three datasets. (b) Illustration of partial AUC on the ROC curve: AUC covers regions $\mathbb{A}+\mathbb{B}+\mathbb{C}$, OPAUC($\alpha+d$) covers $\mathbb{A}+\mathbb{B}$, and WPAUC$(\alpha,d)$ isolates the windowed region $\mathbb{B}$ over the FPR interval $[\alpha, \alpha+d]$.}
    \label{fig:AUC_comparison}
    
\end{figure}

\begin{figure*}[ht]
    \centering
    \includegraphics[width=0.9\linewidth]{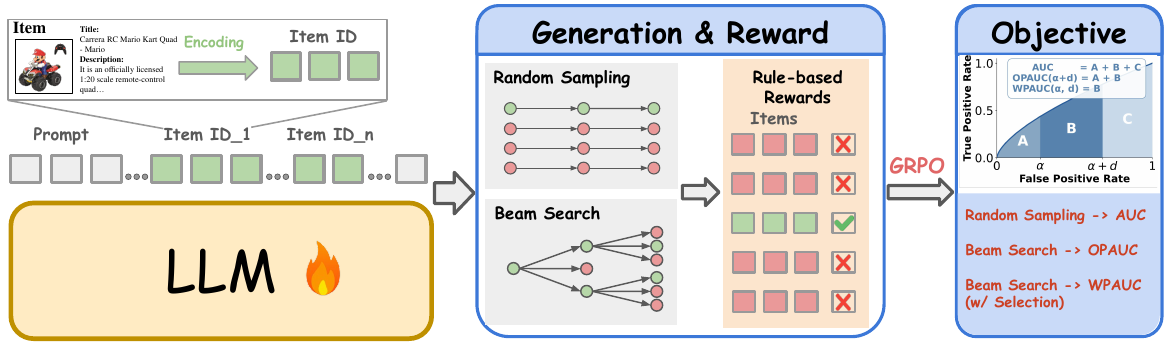}
    \caption{\textbf{Overview of GRPO training for LLM-based recommenders.} The LLM recommender generates candidate item IDs from a history prompt via random sampling or beam search, receives rule-based rewards, and is updated with GRPO. Random sampling induces an AUC objective, while beam-search negatives shift the objective toward One-way Partial AUC (OPAUC). Windowed selection over beam-search samples further yields Windowed Partial AUC (WPAUC).}
    \label{fig:main}
\end{figure*}

% \begin{figure}[ht]
%     \centering
%     \includegraphics[width=0.95\linewidth]{figure/WPAUC.pdf}
%     \caption{(a) OPAUC evaluates the ROC area restricted to the interval FPR $\in [0, d]$. (b) WPAUC generalizes OPAUC by evaluating the area over a shifted interval FPR $\in [\alpha, \alpha+d]$, allowing flexible focus on the low-FPR region.}
%     \label{fig:WPAUC}
% \end{figure}
% Beam search tends to generate negative items that the model assigns high probabilities to, commonly referred to as hard negative items~\citep{vijayakumar2016diverse}

However, the theoretical mechanisms underlying the efficacy of beam search remain underexplored. Existing work generally attributes the advantages of negative items from beam search to their ability to provide more informative training signals~\citep{DBLP:journals/corr/abs-2510-12211, zhou2025onerec}. However, such explanations lack the granularity required for principled algorithmic design. Drawing inspiration from~\citet{DBLP:conf/www/ShiCFZWG023}, we provide a rigorous theoretical investigation of the optimization objective for RL training, yielding two primary insights:
\begin{itemize}[leftmargin=*]
    \item \textbf{Optimizing LLM-based recommenders with GRPO algorithms~\citep{guo2025deepseek} under binary reward is theoretically equivalent to maximizing the AUC metric}. However, AUC is poorly aligned with the Top-$K$ objective in the recommendation task.
    \item \textbf{Incorporating hard negative items generated by beam search into GRPO shifts the objective toward One-way Partial AUC (OPAUC)}, which constrains the false positive rate (FPR) range, as shown in Figure~\ref{fig:AUC_comparison}(b). \textbf{This yields improved alignment with Top-$K$ metrics, following the theoretical analysis in~\citet{DBLP:conf/www/ShiCFZWG023}.}
\end{itemize}

Building on these insights, we provide a theoretical explanation for the role of hard negatives in RL-based LLM recommenders: they implicitly steer the optimization objective toward Top-$K$ metrics. However, the correlation between OPAUC and Top-$K$ metrics with specific $K$ remains insufficient, precluding precise control over the optimization process. To address this limitation, we propose a new metric, \textbf{W}indowed \textbf{P}artial \textbf{AUC} (\textbf{WPAUC}), which constrains the FPR within a windowed interval $[\alpha, \alpha+d]$, as shown in Figure~\ref{fig:AUC_comparison}(b). 
%We demonstrate that WPAUC exhibits a stronger correlation with Top-$K$ metrics and, under certain parameter settings, is equivalent to Recall@$K$.
We formally demonstrate that WPAUC exhibits a superior correlation with Top-$K$ metrics and, under specific parameterization, is equivalent to Recall@$K$. Furthermore, we propose an efficient \textbf{T}hreshold-\textbf{A}djusted \textbf{Win}dowed reweight
method (\textbf{TAWin}), which emphasizes negative samples within the target window while smoothly attenuating those outside it, thereby preserving the Top-$K$ inductive bias during training. Extensive experiments on four real-world datasets validate our theoretical findings and demonstrate state-of-the-art performance. Our primary contributions are as follows:
\begin{itemize}[leftmargin=*]
    % \item We theoretically show that beam search guide the RL optimization objective of LLM recommenders toward Top-$K$ metrics, aligning more closely with the goal of the recommendation task.
    % \item We introduce the novel WPAUC($\alpha,d$) metric, establish its strong alignment with Top-$K$ metrics, and develop an efficient threshold-adjusted windowed reweight RL method \textbf{TAWin} for its optimization.
    % \item Extensive experiments on four real-world datasets confirm our theoretical insights and demonstrate that our method consistently attains state-of-the-art results.
    \item \textbf{Theoretical Analysis}: We demonstrate that beam search implicitly steers the RL optimization of LLM recommenders toward Top-$K$ metrics, aligning the objective more closely with recommendation goals.
    \item \textbf{Methodological Innovation}: We propose
    WPAUC($\alpha,d$) metric to establish strong alignment with Top-$K$ metrics and develop \textbf{TAWin}, an efficient threshold-adjusted windowed reweighting method, for its optimization.
    \item \textbf{Empirical Validation}: Through extensive testing on four real-world datasets, we confirm our theoretical insights and demonstrate that our approach consistently achieves state-of-the-art results.
\end{itemize}

\section{Preliminary}

\subsection{Task Formulation}

We cast LLM-based recommendations as conditional language generation. For each user $u\in\mathcal{U}$ with interaction history $H_u=\{i_1,\ldots,i_n\}$ (items verbalized as text, e.g., titles), the recommender $\pi_\theta$ is trained to generate the \textbf{ground-truth target item $i_t$} under a constrained item vocabulary $\mathcal{I}$~\citep{bao2024decoding, DBLP:journals/corr/abs-2510-12211}. Formally, the decoding objective is:
\begin{equation}
    Y = \mathcal{G}(\pi_{\theta}, H_u, \mathcal{I}), \quad
i = \phi(Y),
\end{equation}
where $\mathcal{G}$ denotes a constrained generation strategy (e.g., random sampling or beam search), and $\phi(\cdot)$ maps the token sequence $Y=(y_1,\ldots,y_{|Y|})$ to a valid item $i$.

\textbf{Constrained Random Sampling} sequentially samples tokens from $\pi_\theta(\cdot\mid y_{<j},H_u)$ while enforcing the valid-token constraint $y_j\in\mathcal{I}$, yielding a stochastic candidate that respects the item vocabulary.

\textbf{Constrained Beam Search} instead performs deterministic decoding by maintaining a beam of the $B$ highest-probability partial hypotheses under the same constraint and returning the $B$ completed sequences (and their mapped items) from the final beam. More implementation details are deferred to Appendix~\ref{appendix:generation_strategy}.

\subsection{Group Relative Policy Optimization (GRPO)}
Group Relative Preference Optimization (GRPO) is a representative reinforcement learning framework for training LLM-based recommenders. At each iteration, it generates a set of $G$ candidate items $\{i_m\}_{m=1}^{G}$ conditioned on each user's history and evaluates them with a verifiable rule-based reward, ensuring alignment with the target item $i_t$:
\begin{equation}
R_{\text{rule}}(i_m, i_t) =
\begin{cases}
1, & \text{if } i_m = i_t, \\
0, & \text{otherwise}.
\end{cases}
\label{eq:rule-reward}
\end{equation}
GRPO normalizes the rewards within each group to compute token-level advantages. Specifically, for the $j$-th token in item $i_m$, the normalized advantage is computed as:
$\hat{A}_{m,j}
=
\frac{r_m - \mathrm{mean}(\{r_m\}_{m=1}^{G})}
     {\mathrm{std}(\{r_m\}_{m=1}^{G})}$
where $r_m$ denotes the rule-based reward $R_{\text{rule}}(i_m, i_t)$, and $\hat{A}_{m,j}$ represents the advantage of the $j$-th token.

The GRPO objective is then optimized via a clipped surrogate function (ignoring the KL regularization term):
\begin{equation}
\begin{aligned}
& \mathcal{J}(\theta)
=
\mathbb{E}_{H_u, \{i_m\} \sim \pi_{old}(\cdot \mid H_u)}
\Bigg[
\frac{1}{G} \sum_{m=1}^{G}
\frac{1}{|Y_m|}
\sum_{j=1}^{|Y_m|}
\Bigg\{ \\
&\min\Bigg[
\rho_{m,j,u} 
\hat{A}_{m,j},
\mathrm{clip}\!\Bigg(
\rho_{m,j,u}, 1-\epsilon,1+\epsilon
\Bigg) \hat{A}_{m,j}
\Bigg] 
\Bigg\}
\Bigg].
\end{aligned}
\label{eq:grpo-obj}
\end{equation}
Here, $\rho_{m,j,u} = \frac{\pi_{\theta}(y_{m,j} \mid y_{m,<j}, H_u)}{\pi_{\text{old}}(y_{m,j} \mid y_{m,<j}, H_u)}$,  $\pi_{\text{old}}$ is the frozen rollout policy used to sample candidates at each epoch. $\epsilon>0$ clips per-step updates.

\subsection{AUC and Partial AUC}
In personalized recommendation, AUC quantifies a model’s ranking capability by assessing the probability that observed positive items are ranked above unobserved negatives for each user. Formally, it can be expressed in a pairwise ranking perspective as:
\begin{equation}
    \text{AUC}_u = \Pr_{i^+ \sim \mathcal{I}_u^+,  i^- \sim \mathcal{I}_u^-} \big[f_{u,i^+} > f_{u,i^-}\big],
\end{equation}
where $\mathcal{I}_u^+$ and $\mathcal{I}_u^-$ are the sets of positive and negative items for user $u$, respectively, and $f_{u,i}$ represents the predicted score of item $i$ for user $u$. In LLM-based recommenders, the score is instantiated as the generation probability assigned to the item tokens:
\begin{equation}
f_{u,i}
=
\prod_{j=1}^{|Y|}
\pi_{\theta}\!\left(y_j \mid y_{<j}, H_u\right),
\quad i = \phi(Y).
\label{eq:fui}
\end{equation}

One-way Partial AUC (OPAUC($\alpha+d$)) restricts evaluation to the low-FPR region (FPR $\le \alpha+d$), thereby emphasizing the ranking quality between positive items and the highest-scored negative items. Formally, the threshold $\eta_{\alpha+d}$ is defined as:
\begin{equation}
\Pr_{i^- \sim \mathcal{I}^-_u}\left[f_{u,i^-} \ge \eta_{\alpha+d}\right] = \alpha+d,
\label{eq:eta_d}
\end{equation}
and only negative items with prediction scores in
$[\eta_{\alpha+d}, 1]$ are included. \citet{DBLP:conf/www/ShiCFZWG023} demonstrate that OPAUC has better alignment to Top-$K$ recommendation performance.

\section{Methodology}
In this section, we first clarify the role of negative sampling in RL training for LLM-based recommenders (\S~\ref{method:negative-sampling}), then propose a new optimization metric WPAUC (\S~\ref{sec:wpauc}), and finally present a smoothed reweight method TAWin (\S~\ref{method;thretopk}).

\subsection{Negative Sampling Meets GRPO}\label{method:negative-sampling}

We formalize how the negative sampling distribution in GRPO implicitly specifies its optimization target. Under a binary reward, GRPO can be reformulated as a pairwise objective equivalent to AUC maximization. Moreover, replacing random sampling with beam search shifts negatives toward the high-score tail of $\mathcal{I}_u^{-}$, inducing an OPAUC objective that better aligns with Top-$K$ ranking metrics, following the analysis of~\citet{DBLP:conf/www/ShiCFZWG023}.

\paragraph{Analysis of GRPO.}
We first establish a formal connection between GRPO and pairwise ranking objectives.

\begin{lemma}[Pairwise Formulation of GRPO]
\label{lem:grpo-to-auc}
Under binary reward defined in Eq.~\eqref{eq:rule-reward} and constrained random sampling strategy, the GRPO objective $\mathcal{J}(\theta)$ in Eq.~\eqref{eq:grpo-obj} can be rewritten as
\begin{equation}
\begin{aligned}
\label{eq:pairwise-form}
\mathcal{J}(\theta)
& =
\mathbb{E}_{H_u}\sqrt{p(H_u)(1-p(H_u))} \cdot \\
&\quad \mathbb{E}_{Y^+ \sim\pi_{old}^+,Y^-\sim\pi_{old}^-}
\Big[
s_\theta^{+}(Y^+,H_u)-s_\theta^{-}(Y^-,H_u)
\Big],
\end{aligned}
\end{equation}
where $p(H_u)\triangleq \Pr_{Y\sim \pi_{old}(\cdot\mid H_u)}\big[r(Y\mid H_u)=1\big]$,
$\pi_{old}^+(\cdot\mid H_u)\triangleq \pi_{old}(\cdot\mid H_u)\mid r(Y\mid H_u)=1$,
$\pi_{old}^-(\cdot\mid H_u)\triangleq \pi_{old}(\cdot\mid H_u)\mid r(Y\mid H_u)=0$, and
\begin{equation}\label{eq:splus-sminus}
s_\theta^{\pm}(Y^\pm,H_u)
\triangleq
\frac{1}{|Y|}\sum_{j=1}^{|Y|}
\begin{cases}
\min(\rho_{j,u},\,1+\epsilon), & (\,+\,),\\
\max(\rho_{j,u},\,1-\epsilon), & (\,-\,).
\end{cases}
\end{equation}
where $\rho_{j,u}=\frac{\pi_\theta(y_j\mid y_{<j}, H_u)}{\pi_{\text{old}}(y_j\mid y_{<j}, H_u)}$.
\end{lemma}
The proof can be found in~\citet{li2025disco}, with a complementary derivation in the Appendix~\ref{appendix:proof_of_lemma1}. 

\paragraph{From sequence-level sampling to item-level pairwise ranking.}

Lemma~\ref{lem:grpo-to-auc} shows that GRPO raises $s_\theta^{+}(Y^{+},H_u)$ while lowering $s_\theta^{-}(Y^{-},H_u)$. Mapping sequences to items via $i^{+}=\phi(Y^{+})$ and $i^-=\phi(Y^{-})$, and noting that $s_\theta^{+}$ and $s_\theta^{-}$ are non-decreasing in the corresponding item scores $f_{u,i^{+}}$ and $f_{u,i^-}$, GRPO therefore induces a pairwise push of the form
\begin{equation}
\label{eq:auc-induced-item}
\Pr_{i^{+}\sim \mathcal{I}_u^{+},\; i^-\sim \mathbb{P}_{\text{neg}}(\cdot\mid u)}
\big[f_{u,i^{+}} > f_{u,i^-}\big],
\end{equation}

where the choice of decoding strategy (e.g., constrained random sampling vs.\ constrained beam search) changes $\mathbb{P}_{\text{neg}}(\cdot\mid u)$ and thus changes the \emph{implicit} pairwise ranking objective optimized by GRPO.

\paragraph{From AUC to OPAUC via Beam Search.}
Compared to AUC, OPAUC($\alpha,d$) restricts attention to top-($\alpha+d$) quantile set of negative items under the policy score $f_{u,i}$:
\begin{equation}
\mathcal{Q}_u^{-}(\alpha+d)\triangleq\{ i^-\in\mathcal{I}_u^- : f_{u,i^-}\ge \eta_{\alpha+d} \},
\label{eq:Q_u}
\end{equation}
where $\eta_{\alpha+d}$ is defined in Eq.~\eqref{eq:eta_d}. This suggests that if the negative sampler is biased toward the upper tail of $f_{u,i^-}$, GRPO will implicitly optimize a \emph{partial} AUC.

\begin{lemma}[Constrained Beam Search and Hard-Negative Selection]
\label{lem:beam-hard-neg}
In the limit as $B \to \infty$, constrained beam search becomes exactly
equivalent to sampling from the top $\eta_{\alpha+d}$ fraction of negatives
ranked by $f_{u,i}$.
\end{lemma}

The proof can be found in Appendix~\ref{appendix:proof_lemma_hard_neg}. Inspired by the lemma, under constrained beam search with finite beam width $B$, the induced negative
sampler concentrates on $\mathcal{Q}_u^{-}(\alpha+d)$ and thus approximates hard negative sampling. We can directly obtain the following proposition.

\begin{proposition}[GRPO with beam negatives optimizes OPAUC]
\label{prop:grpo-opauc}
Consider GRPO training where the positive item is the ground-truth target $i_t\in\mathcal{I}_u^+$ and negatives are generated by constrained beam search.
Under Lemma~\ref{lem:beam-hard-neg}, the induced pairwise objective in Lemma~\ref{lem:grpo-to-auc} becomes
\[
\Pr_{i^+\sim \mathcal{I}_u^+,\,i^-\sim \mathbb{P}_{\text{hard}}(\cdot\mid u)}
\!\big[f_{u,i^+} > f_{u,i^-}\big],
\]
which is equivalent to optimizing $\mathrm{OPAUC}(\alpha+d)$.
\end{proposition}

The direct derivation can be found in Appendix~\ref{appendix:proof_proposition_grpo_opauc}.

% \paragraph{Implication for Top-$K$ Ranking Metrics.}
% Following the analysis in~\citet{DBLP:conf/www/ShiCFZWG023},
% $\mathrm{OPAUC}(d)$ exhibits a stronger connection with Top-$K$ ranking metrics than AUC.
% Consequently, Proposition~\ref{prop:grpo-opauc} provides a theoretical
% justification for the empirical gains of beam search-based GRPO in Top-$K$
% recommendation.

\paragraph{Implication for Top-$K$ alignment.}
By concentrating comparisons on high-score negatives, OPAUC exhibits a tighter connection to Top-$K$ ranking metrics than AUC~\citep{DBLP:conf/www/ShiCFZWG023}.
Therefore, Proposition~\ref{prop:grpo-opauc} provides a theoretical explanation for why beam-search-based GRPO tends to deliver stronger Top-$K$ performance: beam decoding reshapes the induced negative sampling distribution toward the low-FPR region, yielding a Top-$K$-aligned partial-AUC objective.

\subsection{Windowed Partial AUC (WPAUC)}
\label{sec:wpauc}

The above analysis clarifies how beam-search negatives strengthen GRPO by shifting the induced objective toward OPAUC.
However, OPAUC only imposes an upper bound on the FPR, making it hard to optimize a specific $K$.
To enable finer-grained control, we introduce WPAUC, which evaluates ranking quality over a shifted low-FPR window and thereby emphasizes a specific Top-$K$ metric.

\paragraph{Definition.}
Fix a user $u$ with positive set $\mathcal{I}^+_u$ and negative set $\mathcal{I}^-_u$ (with $n^- \coloneqq |\mathcal{I}^-_u|$).
Let $f_{u,i}$ be the predicted score (e.g., the generation probability in Eq.~\eqref{eq:fui}).
Let $\{i^-_{(\sigma)}\}_{\sigma=1}^{n^-}$ denote negative items sorted by descending scores,
$f_{u,i^-_{(1)}} \ge \cdots \ge f_{u,i^-_{(n^-)}}$.
For two FPR bounds $0 \le \alpha < \alpha+d \le 1$, we define the \emph{windowed negative set}
\begin{equation}
\label{eq:window-neg}
\mathcal{W}_u^-(\alpha,d)
\;\coloneqq\;
\Big\{ i^-_{(\sigma)} \in \mathcal{I}_u^- \;:\;
\lceil \alpha n^- \rceil < \sigma \le \lceil (\alpha+d) n^- \rceil
\Big\}.
\end{equation}
We then define WPAUC($\alpha,d$) as:
\begin{equation}
\label{eq:wpauc}
\mathrm{WPAUC}_u(\alpha,d)
\;\coloneqq\;
\Pr_{\,i^+\sim \mathcal{I}_u^+,\, i^-\sim \mathcal{W}_u^-(\alpha,d)}
\!\Big[f_{u,i^+} > f_{u,i^-}\Big].
\end{equation}
Notably, OPAUC is a special case: $\mathrm{OPAUC}(\alpha+d)=\mathrm{WPAUC}(0, \alpha+d)$, and AUC corresponds to $\mathrm{AUC}=\mathrm{WPAUC}(0,1)$.

\paragraph{Why windowing helps Top-$K$ control.}
Unlike OPAUC, which aggregates all negatives in $[0,\alpha+d]$, WPAUC can \emph{shift} the window to concentrate on the specific FPR region that matches a desired $K$.

\begin{theorem}[WPAUC yields tighter Top-$K$ bounds than OPAUC]
\label{thm:wpauc-topk}
Suppose there are $n^+$ positive items and $n^-$ negative items with $n^+<K$ and $n^->K$.
Rank all items in descending order according to scores from any model $f$. Define the FPR window
\begin{equation}
\alpha \;=\; \frac{K-n^+}{n^-},\qquad d \;=\; \frac{n^+}{n^-}.
\label{eq:beta-window}
\end{equation}
Then we have that
\begin{equation}\label{eq:wpauc-recall-bound}
\begin{split}
\frac{\Big\lceil n^+\big(1-\sqrt{1-w}\big)\Big\rceil}{n^+} \ \le\ \text{Recall}@K \ \le\ \frac{\Big\lfloor n^+\sqrt{w}\Big\rfloor}{n^+},
\end{split}
\end{equation}
where $w=\text{WPAUC}(\alpha,d)$. Moreover, for the same $K$, the interval induced by \eqref{eq:wpauc-recall-bound}
is strictly tighter than the corresponding bounds obtained from $\mathrm{OPAUC}(\alpha+d)$.
\end{theorem}

\paragraph{A special case: equivalence to Recall@K with one positive.}
The following lemma formalizes an extreme setting where WPAUC exactly recovers the Recall@$K$ metric.

\begin{lemma}[One-positive equivalence]
\label{lem:wpauc-eq-recall}
Consider a user with exactly one positive item and $n^-$ negatives.
Let $K\ge 2$ and set
$\alpha=\frac{K-1}{n^-}, d=\frac{1}{n^-}.$
Then
\begin{equation}
\label{eq:wpauc-eq-recall}
\mathrm{WPAUC}(\alpha,d) \;=\; \mathrm{Recall@K}.
\end{equation}
\end{lemma}

The complete proof is deferred to Appendix~\ref{app:wpauc-topk}.

%% 图修改，OPAUC-》WPAUC
\begin{figure}[t]
    \centering
    \includegraphics[width=0.95\linewidth]{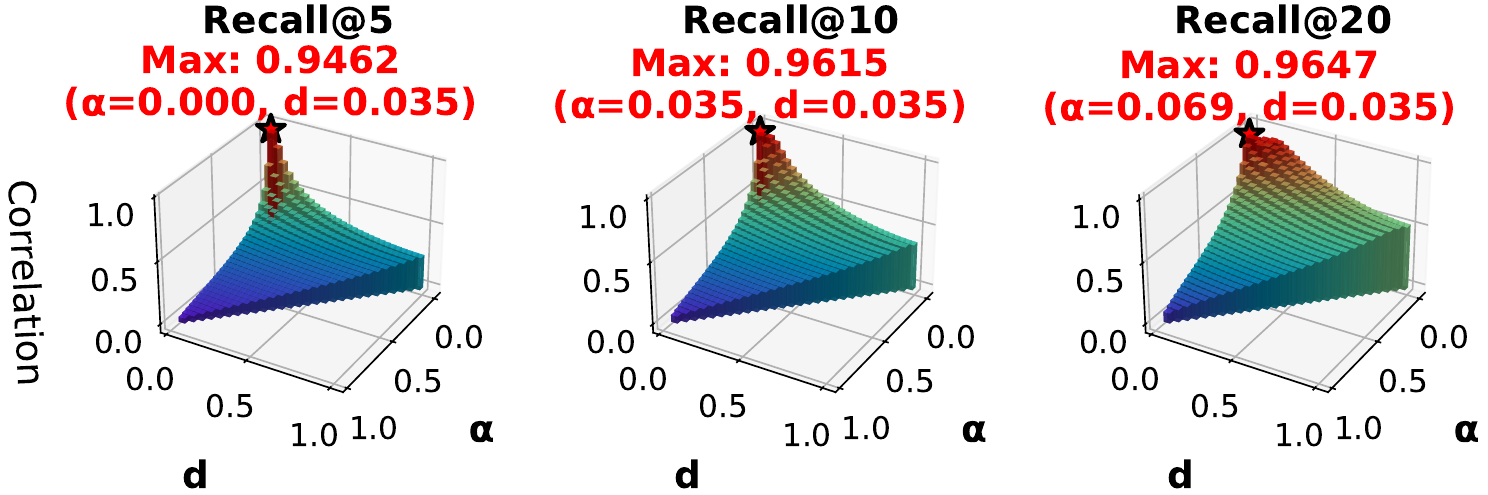}
    \caption{Pearson correlation between Recall@$K$ and WPAUC($\alpha,d$) over parameters $\alpha$ and $d$, for $K\in\{5,10,20\}$. The red star indicates the parameter setting that achieves the maximum correlation in each panel.}
    \label{fig:WPAUC_correlation}
\end{figure}

\paragraph{Empirical Analysis}
We provide simulation evidence to validate the Top-$K$ alignment of WPAUC. We conduct Monte Carlo experiments on synthetic rankings: in each trial, we sample a random permutation containing 10 positives item and 200 negative items, and evaluate both $\mathrm{Recall}@K$ and $\mathrm{WPAUC}(\alpha,d)$ over a grid of window parameters. Repeating this procedure for 10,000 Monte Carlo trials, we compute the resulting Pearson correlation between $\mathrm{Recall}@K$ and $\mathrm{WPAUC}(\alpha,d)$ across trials, and report it in Figure~\ref{fig:WPAUC_correlation}. From the figure, we observe that:
\begin{itemize}[leftmargin=*]
    \item For each $K$, the correlation between WPAUC$(\alpha,d)$ and $\mathrm{Recall}@K$ is maximized at a specific $(\alpha,d)$, indicating that an appropriately chosen low-FPR window yields the strongest Top-$K$ alignment.
    \item The maximizers shift monotonically with $K$: as $K$ increases, the optimal window moves to larger $\alpha$ (i.e., deeper into the negative-score tail), while the optimal window width $d$ remains relatively small. This trend matches the intuition in Theorem~\ref{thm:wpauc-topk}.
\end{itemize}
Overall, these results substantiate that optimizing $\mathrm{WPAUC}(\alpha,d)$ can serve as a practical surrogate for targeted Top-$K$ ranking quality.

\subsection{Threshold-Adjusted Windowed Reweight (TAWin)}\label{method;thretopk}

While WPAUC motivates windowed negative selection, naively discarding negatives outside the target window is undesirable: it wastes sampled negatives (low sample efficiency, especially for costly RL rollouts) and introduces discontinuous inclusion decisions that amplify gradient variance. To address this, we propose \textbf{TAWin}, which replaces hard truncation with a soft threshold-adjusted windowed reweight operator that preserves the inductive bias.

\paragraph{Soft Top-$K$ selection operator.}
To construct TAWin, we first revisit the soft Top-$K$ selection operator.
Given a score vector $\bm{x}\in\mathbb{R}^n$, the hard Top-$K$ operator $T_K(\bm{x})$ returns a $K$-hot indicator vector.
In contrast, a soft Top-$K$ operator defines a continuous mapping
$S_{K,\tau}:\mathbb{R}^n \rightarrow \Delta^{n-1}_K$, where 
\begin{equation}
\Delta^{n-1}_K \triangleq \Big\{\bm{p}\in[0,1]^n \,\big|\, \sum_{i=1}^n p_i = K\Big\},
\end{equation}
where temperature $\tau$ controls smoothness. We require: (i) \textbf{Monotonicity:} if $x_i \ge x_j$, then $[S_{K,\tau}(\bm{x})]_i \ge [S_{K,\tau}(\bm{x})]_j$; (ii) \textbf{Consistency:} $\lim_{\tau\to 0^+} S_{K,\tau}(\bm{x}) = T_K(\bm{x})$.

\paragraph{Threshold-adjusted soft Top-$K$ operator $\mathcal{T}_{K,\tau}$.}
Following~\citet{kexuefm-10373}, we use the clipped exponential function
$g(x)=\min(1,e^x)$ to instantiate the threshold-adjusted soft Top-$K$ selection operator $\mathcal{T}_{K,\tau}$:
\begin{equation}
\begin{split}
\mathcal{T}_{K,\tau}(\bm{x})
&\triangleq
\min\!\Big(\mathbf{1},\exp\!\big(\tfrac{\bm{x}-\lambda(\bm{x})\mathbf{1}}{\tau}\big)\Big),
\\
\text{s.t.}\ \sum_{i=1}^n &\mathcal{T}_{K,\tau}(\bm{x})_i = K,
\label{eq:thretopk-clip-exp}
\end{split}
\end{equation}
where $\mathbf{1}$ denotes the n-dimensional all-ones vector and the threshold $\lambda(\bm{x})$ is the unique scalar that calibrates the total mass to $K$.
In practice, sort $x_{(1)}\ge\cdots\ge x_{(n)}$ and find an index $m<K$ such that $x_{(m)}\ge \lambda \ge x_{(m+1)}$, then $\lambda(\bm{x})$ has the closed form:
\begin{equation}
\lambda(\bm{x})
=
\tau\log\!\Big(\sum_{i=m+1}^{n} \exp\!\big(\tfrac{x_{(i)}}{\tau}\big)\Big)
-\tau\log(K-m),
\label{eq:thretopk-lambda-one-line}
\end{equation}
with $m$ obtained by a short scan over $m\in\{0,\ldots,K-1\}$.
As $\tau\to 0^+$, $\mathcal{T}_{K,\tau}$ becomes sharp and approaches the hard Top-$K$ selector $T_K$.

\begin{table*}[t]
\centering
\caption{Performance comparison under Top-$K$ metrics across four datasets. The best performance is
highlighted in boldface.}
\label{tab:main_table}
\resizebox{\textwidth}{!}{
\begin{tabular}{l ccc ccc ccc ccc}
\toprule
\multirow{2}{*}{Method} & 
\multicolumn{3}{c}{Toys} & 
\multicolumn{3}{c}{Industrial} & 
\multicolumn{3}{c}{Office} & 
\multicolumn{3}{c}{Yelp} \\
\cmidrule(lr){2-4}
\cmidrule(lr){5-7}
\cmidrule(lr){8-10}
\cmidrule(lr){11-13}

& R@1 & R@3 & N@3
& R@1 & R@3 & N@3
& R@1 & R@3 & N@3
& R@1 & R@3 & N@3 \\
\midrule

GRU4Rec
& 0.0090   &  0.0169  & 0.0135
& 0.0461   &  0.0657  & 0.0576
& 0.0384   &  0.0631  & 0.0529
& 0.0074   &  0.0164  & 0.0125\\

Caser
& 0.0125   & 0.0219   & 0.0180
& 0.0371   & 0.0624   & 0.0523
& 0.0450   & 0.0734   & 0.0614
& 0.0079   & 0.0217   & 0.0158 \\

SASRec
& 0.0216   & 0.0359   & 0.0298
& 0.0567   & 0.0761   & 0.0682
& 0.0641   & 0.0923   & 0.0807
& 0.0074   & 0.0175   & 0.0133\\
\midrule

TIGER
& 0.0224   & 0.0383   &  0.0305
& 0.0632   & 0.0852   & 0.0742
& 0.0624   &  0.0986   &  0.0852
& 0.0068   &  0.0154  &  0.0113
\\
LC-Rec
& 0.0253   & 0.0406   & 0.0341
& 0.0727   &  0.0986  & 0.0877
& 0.0900   &  0.1196  & 0.1074
& 0.0094   &  0.0174  & 0.0140 \\

MiniOneRec
& 0.0271   & 0.0458   & 0.0378
& 0.0831   &  0.1100  & 0.0990
& 0.0972   &  0.1219  & 0.1137
& 0.006   &   0.0163 &  0.0120
\\
\midrule
BigRec
& 0.0329   & 0.0510   & 0.0433
& 0.0732   & 0.1012   & 0.0895
& 0.0861   & 0.1201   & 0.1048
& 0.0092   & 0.0187   & 0.0148\\

D3
& 0.0371  & 0.0612   & 0.0512
& 0.0810   & 0.1103   & 0.0980
& 0.0810   & 0.1204   & 0.1040
& 0.0120   & 0.0309   & 0.0228 \\

S-DPO
& 0.0275    & 0.0534   & 0.0449
& 0.0635   &  0.1032  & 0.0906
& 0.0390   & 0.1169   & 0.1033
& 0.0189   & 0.0342   & 0.0395 \\

ReRe
& 0.0411 & 0.0709 & 0.0583
& 0.0783 & 0.1184 & 0.1016
& 0.0830 & 0.1304 & 0.1115
&  0.0206 & 0.0360 & 0.0295
 \\

TAWin
& \textbf{0.0471} & \textbf{0.0761} & \textbf{0.0639} 
& \textbf{0.0904} & \textbf{0.1237} & \textbf{0.1099}
& \textbf{0.0961} & \textbf{0.1341} & \textbf{0.1187}
&  \textbf{0.0227} & \textbf{0.0370} & \textbf{0.0301}
 \\

\bottomrule
\end{tabular}
}
\vspace{-10pt}
\end{table*}

\paragraph{Threshold-Adjusted Windowed Reweight (TAWin)}
Then we construct the TAWin method based on the soft Top-$K$ operator $\mathcal{T}_{K,\tau}$.
For each user $u$, constrained beam search yields $n$ candidate negatives
$\mathcal{Y}_u^-=\{Y_1,\ldots,Y_n\}$, which we rank by model scores $f$.
Let $\sigma_u:\mathcal{Y}_u^-\!\to\!\{1,\ldots,n\}$ map each candidate $Y$ to its rank $\sigma=\sigma_u(Y)$ (higher score $\Rightarrow$ smaller $\sigma$).
Define the normalized rank $\tilde \sigma(\sigma)\triangleq (\sigma-1)/(n-1)$ and an anchor $\sigma_\star$ for the target window's low-FPR boundary.
We assign each candidate a logit by its distance to the anchor:
\begin{equation}
x_u(Y)\;\triangleq\; -\,\big|\tilde \sigma(\sigma_u(Y))-\tilde \sigma(\sigma_\star)\big|.
\end{equation}
Applying $\mathcal{T}_{K,\tau}$ on these logits yields soft windowed weights over the candidates:
\begin{equation}
\bm{w}_u
~=~
\mathcal{T}_{K,\tau}\big([x_u(Y_1),\ldots,x_u(Y_G)]\big),
\label{eq:tawin-weights-rank-map}
\end{equation}
Here, $K$ controls the window width (corresponding to $d$ in WPAUC($\alpha,d$)),
and anchor $\sigma_\star$ determines the window start (corresponding to $\alpha$ in WPAUC($\alpha,d$)).
In the sharp limit (small $\tau$), $\bm{w}_u$ concentrates to a hard window selector $T_K$.

Based on $\bm{w}_u$, we define final weight function $\omega_u(\cdot)$ by
\begin{equation}
\omega_u(Y)\triangleq
\begin{cases}
1, & Y\in\mathcal{Y}_u^+,\\[2pt]
\mathrm{Rescale}\!\big(\bm{w}_u\big)_{\sigma_u(Y)}, & Y\in\mathcal{Y}_u^-,
\end{cases}
\label{eq:tawin-final-weight}
\end{equation}
where $\mathrm{Rescale}(\bm{w}_u)_\sigma=n*w_{u,\sigma}/\sum_{\sigma'=1}^nw_{u,\sigma'}$.

\paragraph{TAWin-Based RL Objective.} TAWin performs sequence level reweighting over rollout candidates by assigning a weight $\omega_u(Y)$ to each sampled sequence $Y$, and scales its policy-gradient contribution to optimize WPAUC($\alpha, d$). This weighting can be applied to multiple RL algorithms (Section~\ref{experiment:rq3}). By default, we use GRPO, yielding: 
\begin{equation}
\begin{aligned}
& \mathcal{J}_{\text{TAWin}}(\theta)
=
\mathbb{E}_{H_u,\, Y_m}
\Bigg[
\frac{1}{n}\sum_{m=1}^{n}\omega_u(Y_m)\;
\frac{1}{|Y_m|}
\sum_{j=1}^{|Y_m|} \\
&\quad \quad \min\Big(
\rho_{m,j,u}\hat A_{m,j},\;
\mathrm{clip}(\rho_{m,j,u},1-\epsilon,1+\epsilon)\hat A_{m,j}
\Big)
\Bigg],
\end{aligned}
\label{eq:tawin-grpo-obj}
\end{equation}
where $\rho_{m,j,u} = \frac{\pi_{\theta}(y_{m,j} \mid y_{m,<j}, H_u)}{\pi_{\text{old}}(y_{m,j} \mid y_{m,<j}, H_u)}$, and the full algorithm is provided in Appendix~\ref{appendix:algorigthm}.

\section{Experiments}
In this section, we evaluate our methods on four public datasets to answer the following questions:

\noindent \textbf{(Q1)} How does TAWin compare with traditional recommenders and LLM recommenders across Top-$K$ metrics?

\noindent \textbf{(Q2)} How does varying TAWin hyperparameters affect performance across different Top-$K$ metrics?

\noindent \textbf{(Q3)} How does TAWin generalize across diverse base models, item encoding strategies, and optimization methods?
    % \item \textbf{(Q4)} How sensitive is TAWin to key hyperparameters and design choices?

\paragraph{Datasets and Metrics.} We evaluate TAWin on four real-world datasets, including three subsets (\textbf{Toys}, \textbf{Industrial}, and \textbf{Office}) from Amazon Review Dataset\footnote{\url{https://cseweb.ucsd.edu/~jmcauley/datasets.html\#amazon_reviews}}~\citep{lakkaraju2013s} and \textbf{Yelp} Dataset\footnote{\url{https://business.yelp.com/data/resources/open-dataset/}}. To assess Top-$K$ performance, we use Recall (R@K) and Normalized Discounted Cumulative Gain (N@K). More details are provided in Appendix~\ref{appendix:dataset} and~\ref{appendix:metrics}.

\paragraph{Baselines.} We consider three categories of baseline methods: (1) Traditional sequential recommendation models, including GRU4Rec~\citep{hidasi2015session}, Caser~\citep{tang2018personalized}, and SASRec~\citep{kang2018self}. (2) Generative recommendation models, represented by TIGER~\citep{DBLP:conf/nips/RajputMSKVHHT0S23}, LC-Rec~\citep{zheng2024adapting}, MiniOneRec~\citep{kong2025minionerec}. (3) LLM-based recommendation models, including BigRec~\citep{DBLP:journals/tors/BaoZWZYLCFT25}, D3~\citep{bao2024decoding}, S-DPO~\citep{DBLP:conf/nips/ChenTZ0SZWC24} and ReRe~\citep{DBLP:journals/corr/abs-2510-12211}. Additional details are provided in Appendix~\ref{appendix:Baselines}.

% Requires: \usepackage{booktabs}
% Optional: \usepackage{array}

\begin{table*}[t]
\centering
\caption{Comparison of ReRe and TAWin across different optimization algorithms. The best results are in bold.}
\label{tab:ablation_dapo_gspo}
\resizebox{0.9\textwidth}{!}{
\begin{tabular}{llccc ccc ccc }
\toprule
& & \multicolumn{3}{c}{Toys} & \multicolumn{3}{c}{Industrial} & \multicolumn{3}{c}{Office}\\
\cmidrule(lr){3-5} \cmidrule(lr){6-8} \cmidrule(lr){9-11}
Algorithm & Method & R@1 & R@3 & N@3 & R@1 & R@3 & N@3 & R@1 & R@3 & N@3 \\
\midrule
\multirow{2}{*}{DAPO}
& ReRe  
& 0.0416 & 0.0711 & 0.0588 
& 0.0814 & 0.1202 & 0.1039 
& 0.0832 & 0.1276 & 0.1098\\
& TAWin  
& \textbf{0.0471} & \textbf{0.0737} & \textbf{0.0626} 
& \textbf{0.0875} & \textbf{0.1248} & \textbf{0.1089} 
& \textbf{0.0910}  & \textbf{0.1317}  & \textbf{0.1149} \\
\midrule
\multirow{2}{*}{GSPO}
& ReRe  
& 0.0404 & 0.0717 & 0.0585
& 0.0736 & 0.1166 & 0.0988 
& 0.0816 & 0.1284 & 0.1094\\
& TAWin  
& \textbf{0.0463} & \textbf{0.0750} & \textbf{0.0631} 
& \textbf{0.0867} & \textbf{0.1224} & \textbf{0.1073} 
& \textbf{0.9556}  & \textbf{0.1305}  & \textbf{0.1163} \\
\bottomrule
\end{tabular}
}
\vspace{-10pt}
\end{table*}

\paragraph{Implementations.} Following~\citet{DBLP:journals/corr/abs-2510-12211}, we train TAWin with a mini-batch size of 512. We use a fixed learning rate of $1\times10^{-5}$ and set the GRPO coefficient $\beta$ to $1\times10^{-3}$. For each input, we sample $n=16$ candidate sequences. We train TAWin for two epochs initialized from the vanilla Qwen2.5-0.5B~\citep{qwen2025qwen25technicalreport} checkpoint on 8 NVIDIA H800 GPUs. Unless otherwise specified, the parameter $K$ in TAWin is set to 4. The anchor $\sigma_\star$ is selected from $\{0,1,2\}$ and the temperature $\tau$ is selected via grid search over $\{1/2,1/3,1/4,1/5,1/6\}$. All LLM-based recommenders use Qwen2.5-0.5B as the backbone.

For ReRe, we adopt beam-search generation with rule-based rewards. The ranking reward variant can be viewed as a special reweighting scheme and its detailed analysis and comparisons are provided in Appendix~\ref{appendix:selection_ablaton}. Additional baseline hyperparameters are deferred to Appendix~\ref{appendix:implement_details}.

\subsection{Overall Experiments (RQ1)}
We evaluate the performance of TAWin across four datasets, with the results reported in Table~\ref{tab:main_table}. From the table, we have the following key insight:
\begin{itemize}[leftmargin=*]
    \item \textbf{Superiority of TAWin across All Baselines.} TAWin consistently establishes a new state-of-the-art, outperforming all baselines on all datasets. Specifically, TAWin achieves average relative improvements of 84.9\%, 52.0\%, and 5.5\% over traditional, generative, and LLM-based recommenders, respectively. We attribute this performance gain to the windowed integration of negative items within the RL framework, which produces a training objective that aligns better with Top-$K$ metrics.
    \item \textbf{Efficacy of LLM-based Recommenders.} LLM-based recommenders demonstrate a clear margin of superiority over traditional sequential recommenders. This improvement stems from their rich semantic representations and extensive world knowledge, thereby strengthening the modeling of user interests.
    \item \textbf{The Advantage of RL in Alignment.} Our results indicate that RL-based methodologies yield consistent improvements over their Supervised Fine-Tuning (SFT) counterparts. For instance, MiniOneRec surpasses TIGER by an average of 22.6\%, while TAWin exceeds BigRec by 49.4\%. This is attributed to RL’s ability to model users’ relative preferences, yielding better alignment with recommendation objectives.
\end{itemize}

\begin{figure}[t]
    \centering
    \includegraphics[width=0.98\linewidth]{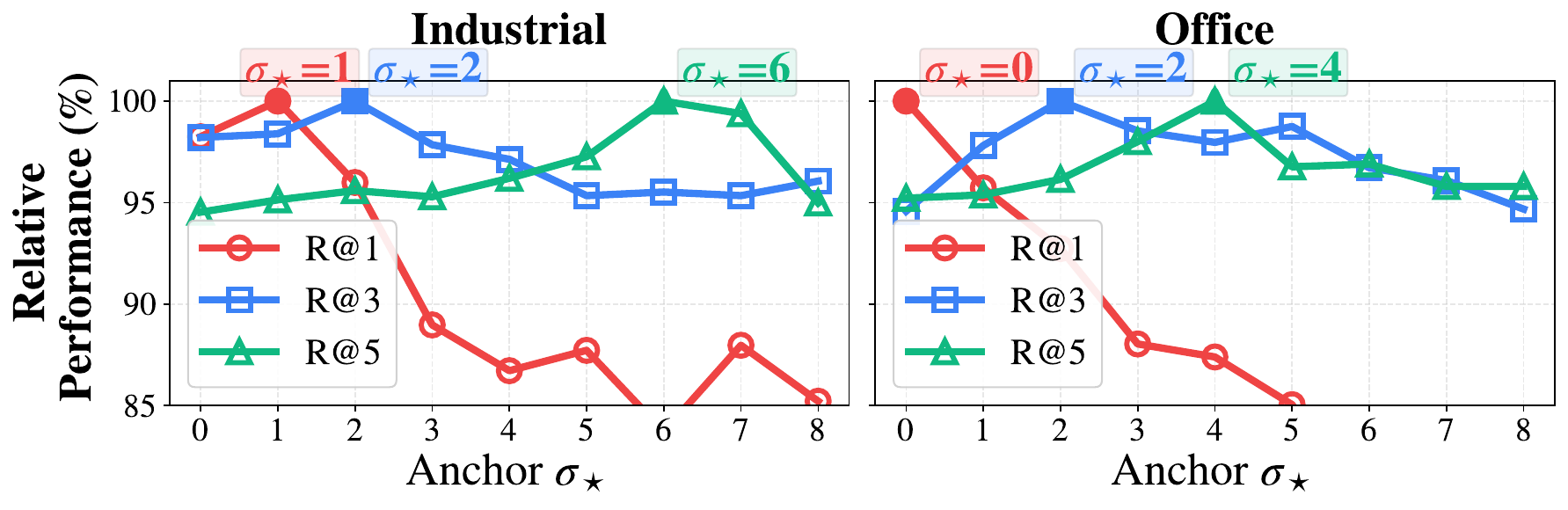}
    \caption{TAWin sensitivity to the anchor $\sigma_\star$ (i.e., $\alpha$ in WPAUC($\alpha,d$)) on two datasets. For each curve, the best-performing setting is highlighted with a filled marker and an annotation indicating the corresponding $\sigma_\star$.}
    \label{fig:WPAUC_hyperparameter}
\end{figure}

\subsection{Effect of Top-$K$ Control in TAWin (RQ2)}
To investigate the controllability of the TAWin, we examine whether modulating its hyperparameters can systematically bias optimization toward specific Top-$K$ metrics. Figure~\ref{fig:WPAUC_hyperparameter} illustrates the relative performance across varying anchor parameter ($\sigma_\star$) on the Industrial and Office datasets. Our analysis reveals several key insights:
\begin{itemize}[leftmargin=*]
    \item \textbf{Unimodal Trend as $\sigma_\star$ Varies.} As shown in Figure~\ref{fig:WPAUC_hyperparameter}, for a fixed target R@K with specific $K$, the relative performance exhibits a  unimodal trend as the anchor $\sigma_\star$ varies. Performance reaches its zenith at a specific anchor value, whereas anchor $\sigma_\star$ that are either too small or overly large result in suboptimal alignment with the targeted Top-$K$ metric. This empirical behavior is consistent with our analytical findings in Section~\ref{sec:wpauc}.
    \item \textbf{Optimal $\sigma_\star$ Increases with Target $K$.} We observe a consistent correlation between the target metric and the optimal anchor value: specifically, smaller target $K$ metrics (e.g., R@1) achieve peak performance at lower anchor values ($\sigma_\star=1$ for Industrial, $\sigma_\star=0$ for Office) compared to larger $K$ metrics like R@5, which peak at higher anchors ($\sigma_\star=6$ and $\sigma_\star=4$, respectively). This trend aligns with our theoretical analysis and simulation results in Section~\ref{sec:wpauc}, indicating that TAWin can flexibly adapt its optimization focus toward different Top-$K$ regimes through appropriate hyperparameter control.
\end{itemize}

\begin{figure}[t]
    \centering
    \includegraphics[width=0.98\linewidth]{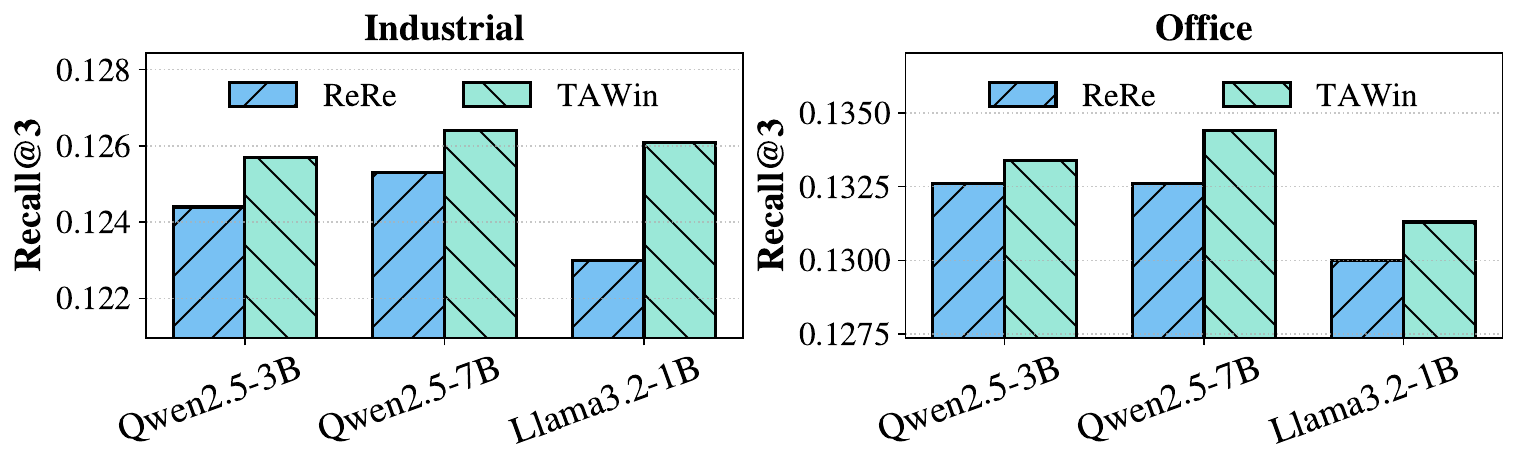}
    \caption{Comparison between ReRe and TAWin across different base models.}
    \label{fig:model_ablation}
\end{figure}

\subsection{Generalization Analysis (RQ3)}\label{experiment:rq3}
In this section, we assess the robustness of our theoretical insights and the empirical efficacy of TAWin. We examine TAWin's versatility across diverse backbone models, optimization algorithms, and item-encoding paradigms.

\paragraph{Generalization across Backbone Models.} We investigate architectural robustness by evaluating models from distinct families and varying scales: Qwen2.5~\citep{qwen2025qwen25technicalreport} (3B and 7B) and Llama-3.2~\citep{grattafiori2024llama3herdmodels} (1B). Figure~\ref{fig:model_ablation} compares the performance of ReRe and TAWin under these configurations. TAWin consistently outperforms ReRe across all backbones, indicating that the benefits of reshaping the RL objective via hard-negative sampling are model-agnostic. Furthermore, the performance gains remain stable as model size increases, suggesting that TAWin possesses favorable scalability.

\paragraph{Generalization across Optimization Algorithms.}
To evaluate generalization across RL optimization algorithms, we extend our evaluation to GSPO~\citep{zheng2025group} and DAPO~\citep{yu2025dapo}. Given that analogous AUC-style reformulations for these algorithms are established in~\citet{li2025disco}, we investigate whether TAWin’s negative-sampling strategy enhances alignment with Top-$K$ recommendation metrics. Table~\ref{tab:ablation_dapo_gspo} presents the performance on the Toys, Industrial, and Office datasets. TAWin consistently surpasses ReRe under both optimization algorithms, demonstrating that the method's improvements are transferable across different RL optimizers.

\begin{table}[t]
\centering
\caption{Performance comparison based on semantic-ID item encoding strategies. The best results are in bold.}
\label{tab:semantic_id_ablation}
\resizebox{0.9\linewidth}{!}{
\begin{tabular}{llccc}
\toprule
Dataset & Method & R@1 & R@3 & N@3 \\
\midrule

\multirow{4}{*}{\textbf{Industrial}}
& MiniOneRec
& & & \\

& \quad - \textit{SFT}
& 0.0726 & 0.0986 & 0.0877
\\
& \quad - \textit{GRPO}
& 0.0831 & 0.1100 & 0.0989
\\
& \quad - \textit{TAWin}
& \textbf{0.0862} & \textbf{0.1158} & \textbf{0.1033}
\\

\midrule

\multirow{4}{*}{\textbf{Office}}
& MiniOneRec
& & & \\
& \quad - \textit{SFT}
& 0.0900 & 0.1196 & 0.1074
\\
& \quad - \textit{GRPO}
& 0.0969 & 0.1325 & 0.1174
\\
& \quad - \textit{TAWin}
& \textbf{0.0974} & \textbf{0.1378} & \textbf{0.1253}
\\

\bottomrule
\end{tabular}
}
\vspace{-5pt}
\end{table}

\paragraph{Generalization across Item Encoding Strategies.} 
Finally, we evaluate generalization beyond title-based encoding by considering semantic item representations, specifically utilizing MiniOneRec~\citep{kong2025minionerec}. We integrate TAWin’s negative sampling strategy into the MiniOneRec training process (denoted as MiniOneRec-TAWin) and compare it against the standard MiniOneRec baseline. As shown in Table~\ref{tab:semantic_id_ablation}, MiniOneRec-TAWin yields consistent improvements on both the Industrial and Office datasets, validating the robustness of TAWin irrespective of the underlying item-encoding strategy.

% \subsection{Ablation Study (RQ4)}

\section{Relation Work}

% Two-Way Partial AUC and Its Properties
% DeepSeekMath: Pushing the Limits of Mathematical Reasoning in Open Language Models

\paragraph{Generative Recommendation}

Driven by rapid advances in large generative models~\citep{DBLP:journals/corr/abs-2303-08774, DBLP:journals/corr/abs-2501-12948}, recommendation is shifting toward generative modeling~\citep{DBLP:conf/recsys/BaoZZWF023, DBLP:conf/icml/ZhaiLLWLCGGGHLS24}, where models directly generate item identifiers. Two main paradigms have emerged: semantic-ID generation with dense item codes~\citep{DBLP:conf/nips/RajputMSKVHHT0S23, DBLP:conf/www/HouHMZ23, DBLP:conf/icde/ZhengHLCZCW24} and LLM-based recommendation framed as natural language generation~\citep{DBLP:journals/corr/abs-2304-10149, DBLP:journals/tois/ZhangXHZLW25}.

Aligning these generative models with user interests is central. Supervised fine-tuning~\citep{DBLP:conf/sigir/00050LH0Z25, DBLP:journals/tors/BaoZWZYLCFT25} is limited by likelihood-based objectives that weakly capture relative preferences, motivating RL-based tuning~\citep{DBLP:conf/nips/ChenTZ0SZWC24, DBLP:journals/corr/abs-2510-12211}. The OneRec series~\citep{DBLP:journals/corr/abs-2502-18965, DBLP:journals/corr/abs-2506-13695, DBLP:journals/corr/abs-2508-20900} shows that SFT followed by RL yields substantial gains, and recent studies further report consistent benefits from beam-search sampling during RL~\cite{DBLP:journals/corr/abs-2510-12211, DBLP:journals/corr/abs-2506-13695}—but the mechanism is unclear. We address this gap with a unified theoretical framework explaining why beam-search sampling improves RL-based tuning.

\paragraph{Negative Sampling for Recommendation}

In discriminative recommender systems, negative sampling is essential for scalable training~\citep{DBLP:journals/corr/abs-2409-07237}. Existing methods fall into static and dynamic schemes: static samplers (e.g., uniform or popularity-based) draw negatives from fixed item distributions~\citep{DBLP:conf/uai/RendleFGS09, DBLP:conf/recsys/Caselles-DupreL18, DBLP:conf/kdd/ChenSSH17}, whereas dynamic samplers (e.g., DNS~\citep{DBLP:conf/sigir/ZhangCWY13}, PRIS~\citep{DBLP:conf/www/Lian0C20}, AdaSIR~\citep{DBLP:conf/www/ChenLJ0C22}) leverage model scores to prioritize hard negatives~\citep{DBLP:conf/wsdm/RendleF14}. Prior theory~\citet{DBLP:conf/www/ShiCFZWG023} shows that, under supervised learning, negative sampling can implicitly steer optimization toward Top-K metrics and improve ranking. Whether this mechanism extends to RL-based generative recommendations remains unclear.

\paragraph{Partial AUC Optimization}
Partial AUC, introduced in~\citet{articleROC}, evaluates ROC performance over restricted FPR/TPR ranges. OPAUC targets low-FPR regimes relevant to drug discovery~\citep{DBLP:journals/fcsc/GaoWHFZ23, DBLP:conf/kdd/NarasimhanA13}, and its connection to Top-K recommendation has been analyzed in~\citet{DBLP:conf/www/ShiCFZWG023}. Subsequent variants include TPAUC~\citep{yang2019two, yang2021all, yang2022optimizing}, which focuses on the upper-left ROC region but aligns weakly with ranking metrics, and Low-Left Partial AUC~\citep{shi2024lower}, which improves Top-K alignment yet suffers from unstable optimization. Motivated by these gaps, we propose WPAUC, a simple and practical objective that more directly aligns with Top-K ranking goals.

\section{Conclusion}
In this work, we provide a principled explanation of why hard negatives benefit RL-trained LLM recommenders: GRPO with binary rewards optimizes AUC, while beam-search negatives reshape the objective toward Partial AUC, improving alignment with Top-K metrics. Building on this insight, we propose WPAUC for controllable Top-K–oriented optimization and TAWin for efficient soft window selection. Extensive experiments across datasets, backbones, item encodings, and RL optimizers validate the theory and show consistent state-of-the-art gains.

\section{Impact Statement}
This work contributes to a deeper understanding of objective design in reinforcement learning–based LLM recommender systems by revealing how negative sampling implicitly reshapes optimization targets. By formalizing the connection between hard negatives, Partial AUC, and Top-K metrics, our analysis provides a principled foundation for developing recommendation algorithms that are better aligned with practical evaluation criteria, potentially leading to more effective and reliable personalization systems.

From an applied perspective, the proposed WPAUC metric and TAWin optimization framework offer practitioners interpretable and controllable tools for tailoring recommendation behavior to specific Top-K regimes, which may improve user satisfaction and system efficiency in large-scale online platforms. Beyond recommendations, our insights into objective shaping via sampling strategies may inform the design of RL objectives in other LLM-based decision-making and ranking tasks.

As with all recommendation technologies, improved personalization may amplify existing issues such as filter bubbles or unequal exposure if deployed without appropriate safeguards. We encourage future work to combine objective-aligned optimization with fairness, diversity, and transparency considerations.

\nocite{langley00}

\bibliography{example_paper}
\bibliographystyle{icml2026}

%%%%%%%%%%%%%%%%%%%%%%%%%%%%%%%%%%%%%%%%%%%%%%%%%%%%%%%%%%%%%%%%%%%%%%%%%%%%%%%
%%%%%%%%%%%%%%%%%%%%%%%%%%%%%%%%%%%%%%%%%%%%%%%%%%%%%%%%%%%%%%%%%%%%%%%%%%%%%%%
% APPENDIX
%%%%%%%%%%%%%%%%%%%%%%%%%%%%%%%%%%%%%%%%%%%%%%%%%%%%%%%%%%%%%%%%%%%%%%%%%%%%%%%
%%%%%%%%%%%%%%%%%%%%%%%%%%%%%%%%%%%%%%%%%%%%%%%%%%%%%%%%%%%%%%%%%%%%%%%%%%%%%%%
\newpage
\appendix
\onecolumn

\section{Generation Strategy}
\label{appendix:generation_strategy}

We provide the formal definitions of the two constrained generation strategies used in
$Y=\mathcal{G}(\pi_{\theta},H_u,\mathcal{I})$.

\paragraph{Tokenization and item constraint.}
Let $\mathcal{V}$ be the base token vocabulary of the LLM, and let $\mathrm{eos}\in\mathcal{V}$ be the end-of-sequence token.
Each valid item $i\in\mathcal{I}$ is verbalized into a unique token sequence via a deterministic serialization function
$\tau:\mathcal{I}\to\mathcal{V}^{+}$, where $\tau(i)=(y_1,\ldots,y_{L_i})$.
We interpret the constrained item vocabulary $\mathcal{I}$ as a \emph{prefix constraint} over token sequences: a partial
hypothesis $y_{1:j}=(y_1,\ldots,y_j)$ is feasible if it is a prefix of at least one serialized item.

Formally, define the set of \emph{admissible next tokens} after a prefix $y_{<j}=(y_1,\ldots,y_{j-1})$ as
\begin{equation}
\label{eq:admissible-set}
\mathcal{A}(y_{<j};\mathcal{I})
\;=\;
\Big\{ y\in\mathcal{V} \;\Big|\; \exists\, i\in\mathcal{I}\ \text{s.t.}\ \tau(i) \text{ has prefix } (y_{<j},y) \Big\}.
\end{equation}
In practice, $\mathcal{A}(\cdot;\mathcal{I})$ can be implemented efficiently with a prefix trie built over $\{\tau(i): i\in\mathcal{I}\}$.

The mapping $\phi(\cdot)$ in the main text is then defined by exact match:
\begin{equation}
\label{eq:phi-def}
\phi(Y)=i
\quad\Longleftrightarrow\quad
\tau(i)=Y,
\end{equation}
which is well-defined when serialization is unique.

\paragraph{Masked (constrained) token distribution.}
Given the model $\pi_{\theta}$, history $H_u$, and prefix $y_{<j}$, define the constrained distribution over the next token by
masking invalid tokens and renormalizing:
\begin{equation}
\label{eq:masked-dist}
\tilde{\pi}_{\theta}(y \mid y_{<j}, H_u, \mathcal{I})
\;=\;
\frac{\pi_{\theta}(y \mid y_{<j}, H_u)\cdot \mathbf{1}\!\left[y\in \mathcal{A}(y_{<j};\mathcal{I})\right]}
{\sum\limits_{y'\in \mathcal{A}(y_{<j};\mathcal{I})} \pi_{\theta}(y' \mid y_{<j}, H_u)}.
\end{equation}
When $\mathcal{A}(y_{<j};\mathcal{I})$ is constructed from a trie over $\{\tau(i)\}$, the denominator is guaranteed to be nonzero
for any feasible prefix.

\subsection{Constrained Random Sampling}
\label{app:random-sampling}

\paragraph{Definition.}
Constrained random sampling generates a stochastic sequence by sampling each token sequentially from the masked distribution:
\begin{align}
\label{eq:crs}
y_j \;&\sim\; \tilde{\pi}_{\theta}(\cdot \mid y_{<j}, H_u, \mathcal{I}),
\qquad j=1,2,\ldots,T, \\
Y \;&=\; (y_1,\ldots,y_T),
\qquad \text{where } T=\min\{j: y_j=\mathrm{eos}\}\ \text{(or a max length cutoff)}.
\end{align}
Equivalently, the induced sampling distribution over full sequences is
\begin{equation}
\label{eq:crs-joint}
\Pr_{\mathcal{G}_{\mathrm{RS}}}(Y \mid H_u,\mathcal{I})
\;=\;
\prod_{j=1}^{|Y|}
\tilde{\pi}_{\theta}(y_j \mid y_{<j}, H_u, \mathcal{I}),
\end{equation}
and the decoded item is $i=\phi(Y)$.

\paragraph{Output.}
For single-item recommendation, $\mathcal{G}_{\mathrm{RS}}$ returns one $Y$ (and thus one item $\phi(Y)$).
If multiple candidates are needed (e.g., for negative sampling), repeat Eq.~\eqref{eq:crs} independently to obtain a multiset of samples.

\subsection{Constrained Beam Search}
\label{app:beam-search}

\paragraph{Prefix scoring.}
For any partial hypothesis $y_{1:j}$, define its cumulative log-likelihood score under $\pi_{\theta}$ as
\begin{equation}
\label{eq:prefix-score}
h_{\theta}(y_{1:j}\,;\,H_u)
\;=\;
\sum_{t=1}^{j} \log \pi_{\theta}(y_t \mid y_{<t}, H_u).
\end{equation}

\paragraph{Beam update under the constraint.}
Let $B$ be the beam width. Beam search maintains a set of active prefixes $\mathcal{B}_j$ after $j$ decoding steps.
Initialize $\mathcal{B}_0=\{()\}$ with the empty prefix and an empty set of completed hypotheses $\mathcal{F}=\emptyset$.
At step $j\ge 1$, expand each active prefix $y_{<j}\in\mathcal{B}_{j-1}$ only with admissible tokens:
\begin{equation}
\label{eq:beam-expand}
\mathcal{C}_j
=
\bigcup_{y_{<j}\in \mathcal{B}_{j-1}}
\Big\{ (y_{<j}, y)\;:\; v\in \mathcal{A}(y_{<j};\mathcal{I}) \Big\}.
\end{equation}
Select the top-$B$ candidates by score to form the next beam:
\begin{equation}
\label{eq:beam-topb}
\mathcal{B}_j
=
\operatorname{TopB}\big(\mathcal{C}_j;\ h_{\theta}(\cdot\,;\,H_u)\big),
\end{equation}
where $\operatorname{TopB}(\cdot)$ returns the $B$ highest-scoring sequences.
Any hypothesis that ends with $\mathrm{eos}$ and exactly matches a valid item serialization is moved to the completed set:
\begin{equation}
\label{eq:beam-finish}
\mathcal{F}
\leftarrow
\mathcal{F}\ \cup\ \{ y\in\mathcal{B}_j : y_{|Y|}=\mathrm{eos} \ \wedge\ \exists\, i\in\mathcal{I}\ \text{s.t.}\ \tau(i)=Y \}.
\end{equation}
Decoding stops when either (i) a maximum length $L_{\max}$ is reached, or (ii) all beams are completed.

\paragraph{Returning Top-$B$ sequences and mapped items.}
Let $\operatorname{TopB}(\mathcal{F})$ denote the top-$B$ completed sequences in $\mathcal{F}$ by $h_{\theta}(\cdot\,;\,H_u)$.
Constrained beam search returns
\begin{equation}
\label{eq:return-topk}
\{Y^{(1)},\ldots,Y^{(B)}\} = \operatorname{TopB}\big(\mathcal{F};\ h_{\theta}(\cdot\,;\,H_u)\big),
\qquad
i^{(b)}=\phi\!\left(Y^{(b)}\right).
\end{equation}

The algorithm can be found in Algorithm~\ref{alg:constrained-beam}.

\begin{algorithm}[tb]
  \caption{Constrained Beam Search $\mathcal{G}_{\mathrm{BS}}(\pi_{\theta},H_u,\mathcal{I})$}
  \label{alg:constrained-beam}
  \begin{algorithmic}
    \STATE {\bfseries Input:} model $\pi_{\theta}$, history $H_u$, item set $\mathcal{I}$ (via trie), beam width $B$, return size $B$, max length $L_{\max}$
    \STATE {\bfseries Output:} Top-$B$ decoded sequences $\{Y^{(b)}\}_{b=1}^{B}$ and mapped items $\{i^{(b)}=\phi(Y^{(b)})\}_{b=1}^{B}$
    \STATE Initialize beam $\mathcal{B}_0 \leftarrow \{()\}$ and completed set $\mathcal{F}\leftarrow\emptyset$
    \FOR{$j=1$ {\bfseries to} $L_{\max}$}
      \STATE Initialize candidate set $\mathcal{C}_j \leftarrow \emptyset$
      \FORALL{$y_{<j}\in \mathcal{B}_{j-1}$}
        \STATE Compute admissible token set $\mathcal{A}(y_{<j};\mathcal{I})$ (Eq.~\eqref{eq:admissible-set})
        \FORALL{$v \in \mathcal{A}(y_{<j};\mathcal{I})$}
          \STATE Add candidate $(y_{<j},y)$ into $\mathcal{C}_j$
        \ENDFOR
      \ENDFOR
      \STATE Set $\mathcal{B}_j \leftarrow \operatorname{TopB}(\mathcal{C}_j;\ h_{\theta}(\cdot\,;\,H_u))$ using the score in Eq.~\eqref{eq:prefix-score}
      \FORALL{$y \in \mathcal{B}_j$}
        \IF{$y_{|y|}=\mathrm{eos}$ {\bfseries and} $\exists\, i\in\mathcal{I}$ s.t. $\tau(i)=Y$}
          \STATE $\mathcal{F} \leftarrow \mathcal{F} \cup \{Y\}$
        \ENDIF
      \ENDFOR
      \IF{the stopping criterion is satisfied}
        \STATE {\bfseries break}
      \ENDIF
    \ENDFOR
    \STATE Return $\{Y^{(b)}\}_{b=1}^{B} = \operatorname{TopB}(\mathcal{F};\ h_{\theta}(\cdot\,;\,H_u))$ and $i^{(b)}=\phi(Y^{(b)})$ for $b=1,\ldots,B$
  \end{algorithmic}
\end{algorithm}

\section{Proof of Lemma~\ref{lem:grpo-to-auc}}
\label{appendix:proof_of_lemma1}

\begin{proof}
We focus on the reward-driven term in Eq.~\eqref{eq:grpo-obj} and omit the KL regularizer since it does not depend on the binary reward outcomes.
Under constrained random sampling, the $G$ rollouts $\{Y_m\}_{m=1}^G$ are i.i.d. from $\pi_{old}(\cdot\mid H_u)$, hence by symmetry the outer averaging over $m$ satisfies
\begin{equation}
\mathbb{E}_{\{Y_m\}\sim \pi_{old}(\cdot\mid H_u)}
\Big[\frac{1}{G}\sum_{m=1}^G F(Y_m)\Big]
=
\mathbb{E}_{Y\sim \pi_{old}(\cdot\mid H_u)}[F(Y)].
\end{equation}
Therefore, it suffices to analyze a generic rollout $Y\sim \pi_{old}(\cdot\mid H_u)$.

Let the binary reward for a rollout be
\begin{equation}
r(Y\mid H_u)\;\triangleq\;R_{\text{rule}}(\phi(Y), i_t)\in\{0,1\}.
\end{equation}
Define
\begin{equation}
p(H_u)\;\triangleq\;\mathbb{E}_{Y\sim \pi_{old}(\cdot\mid H_u)}[r(Y\mid H_u)]
=\Pr_{Y\sim \pi_{old}(\cdot\mid H_u)}\big[r(Y\mid H_u)=1\big],
\end{equation}
so that $\mathrm{Var}(r(Y\mid H_u))=p(H_u)\big(1-p(H_u)\big)$.
Following the standard population-normalization analysis (which is exact when the empirical mean/std are replaced by their expectations, and becomes accurate as $G$ grows), the normalized advantage can be written as
\begin{equation}
A(Y\mid H_u)
\;\triangleq\;
\frac{r(Y\mid H_u)-p(H_u)}{\sqrt{p(H_u)\big(1-p(H_u)\big)}}
=
\begin{cases}
\sqrt{\dfrac{1-p(H_u)}{p(H_u)}}, & r(Y\mid H_u)=1, \\[6pt]
-\sqrt{\dfrac{p(H_u)}{1-p(H_u)}}, & r(Y\mid H_u)=0.
\end{cases}
\label{eq:adv-pop}
\end{equation}
Moreover, since $r(Y\mid H_u)$ is item-level, the same advantage applies to every token in the sequence, i.e., $\hat{A}_{m,j}=A(Y_m\mid H_u)$.

Now consider the clipped surrogate inside Eq.~\eqref{eq:grpo-obj}. For a token ratio
\begin{equation}
x_{j}(Y,H_u)
\;\triangleq\;
\frac{\pi_\theta(y_j\mid y_{<j}, H_u)}{\pi_{\text{old}}(y_j\mid y_{<j}, H_u)},
\end{equation}
the GRPO token contribution is
\begin{equation}
\min\Big[x_j(Y,H_u)\,A(Y\mid H_u),\ \mathrm{clip}(x_j(Y,H_u),1-\epsilon,1+\epsilon)\,A(Y\mid H_u)\Big].
\end{equation}
When $A(Y\mid H_u)>0$, the above equals
\begin{equation}
A(Y\mid H_u)\cdot \min\big(x_j(Y,H_u),1+\epsilon\big)
\;=\;
A(Y\mid H_u)\cdot g^{+}\!\big(x_j(Y,H_u),1\big),
\end{equation}
and when $A(Y\mid H_u)\le 0$, it equals
\begin{equation}
-|A(Y\mid H_u)|\cdot \max\big(x_j(Y,H_u),1-\epsilon\big)
\;=\;
-|A(Y\mid H_u)|\cdot g^{-}\!\big(x_j(Y,H_u),1\big),
\end{equation}
where $g^{+}(x,1)=\min(x,1+\epsilon)$ and $g^{-}(x,1)=\max(x,1-\epsilon)$.

Putting these together, the reward-driven objective becomes
\begin{equation}
\begin{aligned}
\mathcal{J}(\theta)
&=
\mathbb{E}_{H_u}\mathbb{E}_{Y\sim \pi_{old}(\cdot\mid H_u)}
\Bigg[
\frac{1}{|Y|}\sum_{j=1}^{|Y|}
\Big(
\mathbb{I}\{r(Y\mid H_u)=1\}A(Y\mid H_u)\,g^+(x_j(Y,H_u),1) \\
&\hspace{5.6cm}
-\mathbb{I}\{r(Y\mid H_u)=0\}|A(Y\mid H_u)|\,g^-(x_j(Y,H_u),1)
\Big)
\Bigg].
\end{aligned}
\end{equation}
Conditioning on the two reward events and using \eqref{eq:adv-pop}, we obtain
\begin{equation}
\begin{aligned}
\mathcal{J}(\theta)
&=
\mathbb{E}_{H_u}\Bigg[
p(H_u)\cdot \sqrt{\frac{1-p(H_u)}{p(H_u)}}\;
\mathbb{E}_{Y^+\sim \pi_{old}^+(\cdot\mid H_u)}
\frac{1}{|Y^+|}\sum_{j=1}^{|Y^+|} g^+\!\big(x_j(Y^+,H_u),1\big) \\
&\quad\quad
-(1-p(H_u))\cdot \sqrt{\frac{p(H_u)}{1-p(H_u)}}\;
\mathbb{E}_{Y^-\sim \pi_{old}^-(\cdot\mid H_u)}
\frac{1}{|Y^-|}\sum_{j=1}^{|Y^-|} g^-\!\big(x_j(Y^-,H_u),1\big)
\Bigg] \\
&=
\mathbb{E}_{H_u}\sqrt{p(H_u)\big(1-p(H_u)\big)}
\Bigg[
\mathbb{E}_{Y^+\sim \pi_{old}^+(\cdot\mid H_u)} s_\theta^+(Y^+,H_u)
-\mathbb{E}_{Y^-\sim \pi_{old}^-(\cdot\mid H_u)} s_\theta^-(Y^-,H_u)
\Bigg] \\
&=
\mathbb{E}_{H_u}\sqrt{p(H_u)\big(1-p(H_u)\big)}\;
\mathbb{E}_{Y^+\sim \pi_{old}^+(\cdot\mid H_u),\,Y^-\sim \pi_{old}^-(\cdot\mid H_u)}
\big[s_\theta^+(Y^+,H_u)-s_\theta^-(Y^-,H_u)\big],
\end{aligned}
\end{equation}
which is exactly Eq.~\eqref{eq:pairwise-form}. This completes the proof.
\end{proof}

\section{Proof of Lemma~\ref{lem:beam-hard-neg}}
\label{appendix:proof_lemma_hard_neg}

\begin{proof}
Fix a user history $H_u$ and consider constrained decoding over the finite item vocabulary $\mathcal{I}$, where each valid item $i\in\mathcal{I}$ corresponds to a unique constrained token sequence $Y$ with $\phi(Y)=i$.
Under the policy $\pi_{\theta}(\cdot\mid H_u)$, the (policy) score of item $i$ is
\begin{equation}
f_{u,i}
\;\triangleq\;
\prod_{j=1}^{|Y(i)|}\pi_{\theta}\!\left(y_j(i)\mid y_{<j}(i), H_u\right),
\end{equation}
which is exactly the probability of generating $Y(i)$ under constrained generation.

Constrained beam search with beam size $B$ keeps, at each decoding step, the $B$ highest-probability partial hypotheses (prefixes) under $\pi_{\text{old}}(\cdot\mid H_u)$, and finally returns the Top-$K$ completed sequences in the terminal beam. Because the constrained item vocabulary is finite, the set of valid completed sequences
\begin{equation}
\mathcal{Y}(H_u)\;\triangleq\;\{Y: \phi(Y)=i, i\in\mathcal{I}\}
\end{equation}
is finite. Let $N_j$ denote the number of valid partial prefixes of length $j$ consistent with the constraint. Then for any fixed depth $j$, if $B\ge N_j$, beam search retains \emph{all} valid prefixes of length $j$ (since there are at most $N_j$ such prefixes). Consequently, in the limit $B\to\infty$, for every depth $j$ the beam contains all valid prefixes, and thus at termination the search enumerates all valid completed sequences in $\mathcal{Y}(H_u)$ and ranks them exactly by their probabilities.

Let $\mathcal{I}_u^-$ be the negative-item set for user $u$, with $n_u^- \triangleq |\mathcal{I}_u^-|$, and sort negatives by score in descending order:
\begin{equation}
f_{u,i_{(1)}^-} > f_{u,i_{(2)}^-} > \cdots > f_{u,i_{(n_u^-)}^-},
\end{equation}
In the limit $B\to\infty$, the Top-$K$ completed sequences returned by constrained beam search correspond exactly to the Top-$K$ items
\begin{equation}
\{i_{(1)}^-,\ldots,i_{(K)}^-\}
\quad \text{ranked by } f_{u,i}.
\end{equation}

Now set $\alpha+d \triangleq K/n_u^-$ and let $\eta_{\alpha+d}$ be the OPAUC threshold defined by Eq.~\eqref{eq:eta_d}, i.e.,
$\Pr_{i^-\sim \mathcal{I}_u^-}\!\big[f_{u,i^-}\ge \eta_{\alpha+d}\big]=\alpha+d$.
Under the no-tie assumption (or after tie breaking), this implies
\begin{equation}
\mathcal{Q}_u^-(\alpha+d)
=
\{i^-\in\mathcal{I}_u^-: f_{u,i^-}\ge \eta_{\alpha+d}\}
=
\{i_{(1)}^-,\ldots,i_{(K)}^-\}.
\end{equation}
Therefore, as $B\to\infty$, constrained beam search selects negatives \emph{exactly} from the top-$d$ quantile set $\mathcal{Q}_u^-(\alpha+d)$ (equivalently, those with $f_{u,i^-}\ge \eta_{\alpha+d}$). In particular, if one defines the induced negative sampler $\mathbb{P}_{\text{neg}}(\cdot\mid u)$ by uniformly drawing from the returned Top-$K$ beam outputs, then $\mathbb{P}_{\text{neg}}(\cdot\mid u)$ is precisely the uniform distribution over $\mathcal{Q}_u^-(\alpha+d)$ in the $B\to\infty$ limit. This establishes the claimed equivalence.
\end{proof}

\section{Proof of Proposition~\ref{prop:grpo-opauc}}
\label{appendix:proof_proposition_grpo_opauc}

\begin{proof}
By Lemma~\ref{lem:grpo-to-auc}, GRPO induces the pairwise ranking push
\begin{equation}
\Pr_{i^{+}\sim \mathcal{I}_u^{+},\; i^-\sim \mathbb{P}_{\text{neg}}(\cdot\mid u)}
\big[f_{u,i^{+}} > f_{u,i^-}\big],
\end{equation}
where $\mathbb{P}_{\text{neg}}(\cdot\mid u)$ is the item-level negative sampler determined by the decoding strategy.

When negatives are generated by constrained beam search, Lemma~\ref{lem:beam-hard-neg} implies that (in the hard-negative limit) the induced sampler equals the distribution over the top-$d$ quantile negative set,
\begin{equation}
\mathbb{P}_{\text{neg}}(\cdot\mid u)\;=\;\mathbb{P}_{\text{hard}}(\cdot\mid u),
\qquad
\mathrm{supp}\big(\mathbb{P}_{\text{hard}}(\cdot\mid u)\big)\subseteq \mathcal{Q}_u^-(\alpha+d),
\end{equation}
with $\mathcal{Q}_u^-(\alpha+d)=\{i^-\in\mathcal{I}_u^-: f_{u,i^-}\ge \eta_{\alpha+d}\}$ defined in Eq.~\eqref{eq:Q_u}. Substituting this into the induced pairwise objective yields
\begin{equation}
\Pr_{i^+\sim \mathcal{I}_u^+,\,i^-\sim \mathbb{P}_{\text{hard}}(\cdot\mid u)}
\big[f_{u,i^+} > f_{u,i^-}\big].
\end{equation}
By definition of $\eta_{\alpha+d}$ in Eq.~\eqref{eq:eta_d}, sampling $j\sim \mathbb{P}_{\text{hard}}(\cdot\mid u)$ is equivalent to sampling negatives conditional on belonging to the upper-tail event $\{f_{u,i^-}\ge \eta_{\alpha+d}\}$, i.e., restricting comparisons to $\mathcal{Q}_u^-(\alpha+d)$. Therefore the above pairwise probability is exactly the one-way partial AUC at level $d$, namely $\mathrm{OPAUC}(\alpha+d)$. This proves the claim.
\end{proof}

\section{Proof of Theorem~\ref{thm:wpauc-topk}}
\label{app:wpauc-topk}

\begin{proof}
Let $n\coloneqq n^+$ and let $i$ be the number of positives in the Top-$K$ prefix, so $\mathrm{Recall@K}=i/n$.
Set
\begin{equation}
\alpha=\frac{K-n}{n^-},\qquad d=\frac{n}{n^-},
\end{equation}
then (up to the ceiling operators) $\mathcal{W}^-(\alpha,d)$ contains exactly the negatives with ranks
\begin{equation}
K-n < r \le K,
\end{equation}
hence $|\mathcal{W}^-(\alpha,d)|=n$.
Moreover, the top $K-n$ negatives $\{i^-_{(r)}\}_{r\le K-n}$ \emph{always} appear before all window negatives in score order, so the Top-$K$ event is fully determined by the relative order of the $2n$ items in
$\mathcal{I}^+\cup \mathcal{W}^-(\alpha,d)$:
among these $2n$ items, exactly $i$ positives must fall in the top $n$ positions.

Under this condition, $\mathrm{WPAUC}(\alpha,d)$ is simply the fraction of correctly ordered $(+,-)$ pairs within this
$+\!/-$ sequence of length $2n$ (where $-$ denotes a window negative). We now bound it by extremal arrangements.

\paragraph{Maximum.}
The maximum $\mathrm{WPAUC}(\alpha,d)$ given $i$ is attained when the first half contains $i$ leading positives:
\[
\underbrace{+\cdots+}_{i}\ \underbrace{-\cdots-}_{n-i}\ \Big|\ 
\underbrace{+\cdots+}_{n-i}\ \underbrace{-\cdots-}_{i}.
\]
Then the number of concordant pairs is $in+(n-i)i=i(2n-i)$, hence
\begin{equation}
\mathrm{WPAUC}(\alpha,d)\ \le\ \frac{i(2n-i)}{n^2}.
\end{equation}

\paragraph{Minimum.}
The minimum is attained when the first half contains $i$ trailing positives:
\[
\underbrace{-\cdots-}_{n-i}\ \underbrace{+\cdots+}_{i}\ \Big|\ 
\underbrace{-\cdots-}_{i}\ \underbrace{+\cdots+}_{n-i},
\]
yielding exactly $i^2$ concordant pairs, hence
\begin{equation}
\mathrm{WPAUC}(\alpha,d)\ \ge\ \frac{i^2}{n^2}.
\end{equation}

Let $w\coloneqq \mathrm{WPAUC}(\alpha,d)$. Since $i$ is integer, inverting the above gives
\begin{equation}
\Big\lceil n\big(1-\sqrt{1-w}\big)\Big\rceil \ \le\ i \ \le\ \Big\lfloor n\sqrt{w}\Big\rfloor,
\end{equation}
and dividing by $n$ yields \eqref{eq:wpauc-recall-bound}.

\paragraph{Tightness vs.\ $\mathrm{OPAUC}(\alpha{+}d)$.}
Note that $\alpha{+}d=\frac{K}{n^-}$, hence $\mathrm{OPAUC}(\alpha{+}d)$ uniformly samples negatives from the top-$K$ negative
prefix $\{i^-_{(r)}:1\le r\le K\}$.
With $\alpha,d$ in \eqref{eq:beta-window}, split this prefix into two disjoint blocks
\[
\mathcal{H}^- \coloneqq \{i^-_{(r)}:1\le r\le K-n\},\qquad
\mathcal{W}^- \coloneqq \{i^-_{(r)}:K-n< r\le K\}=\mathcal{W}^-(\alpha,d),
\]
where $|\mathcal{H}^-|=K-n$ and $|\mathcal{W}^-|=n$.
Define
\[
A \coloneqq \Pr_{i^+\sim\mathcal{I}^+,\,j\sim\mathcal{H}^-}\!\big[f(i^+)>f(j)\big],\qquad
B \coloneqq \Pr_{i^+\sim\mathcal{I}^+,\,j\sim\mathcal{W}^-}\!\big[f(i^+)>f(j)\big]
= \mathrm{WPAUC}(\alpha,d).
\]
By uniform mixing over $\mathcal{H}^-\cup\mathcal{W}^-$, we have the exact decomposition
\begin{equation}\label{eq:opauc-mix}
\mathrm{OPAUC}(\alpha{+}d)
\;=\;
\frac{K-n}{K}\,A \;+\; \frac{n}{K}\,B.
\end{equation}
Moreover, since every hard negative has score no smaller than every window negative
($f(i^-_{(r)})\ge f(i^-_{(s)})$ for $r\le K-n < s\le K$), for any fixed $i^+$,
\(
\mathbb{1}\{f(i^+)>j\}=1 \text{ for } j\in\mathcal{H}^-
\Rightarrow
\mathbb{1}\{f(i^+)>j'\}=1 \text{ for } j'\in\mathcal{W}^-,
\)
and averaging yields the monotone constraint
\begin{equation}\label{eq:A-le-B}
0\le A \le B \le 1.
\end{equation}

Let $o\coloneqq \mathrm{OPAUC}(\alpha{+}d)\in(0,1)$.
From \eqref{eq:opauc-mix}--\eqref{eq:A-le-B},
\[
o \le \frac{K-n}{K}B+\frac{n}{K}B = B
\quad\Rightarrow\quad B\ge o,
\qquad
o \ge \frac{n}{K}B
\quad\Rightarrow\quad B\le \min\!\Big\{1,\frac{K}{n}o\Big\}.
\]
Hence $\mathrm{OPAUC}(\alpha{+}d)$ only localizes the window statistic $B$ to a \emph{non-degenerate} interval
\begin{equation}\label{eq:B-interval-from-o}
B \in \Big[o,\ \min\!\big\{1,\tfrac{K}{n}o\big\}\Big],
\end{equation}
whose length is strictly positive for $o\in(0,1)$ since $K>n$.

Now denote the WPAUC-induced Recall bounds in \eqref{eq:wpauc-recall-bound} by
\[
\underline R(B)\ \le\ \mathrm{Recall@K}\ \le\ \overline R(B),
\]
where $\underline R(\cdot)$ and $\overline R(\cdot)$ are monotone nondecreasing in $B$ (as proved above).
If we only observe $o=\mathrm{OPAUC}(\alpha{+}d)$, then $B$ may take any value in \eqref{eq:B-interval-from-o}, so the set of
$\mathrm{Recall@K}$ consistent with $o$ must contain
\[
\bigcup_{B\in [\,o,\ \min\{1,\frac{K}{n}o\}\,]}\ [\underline R(B),\overline R(B)]
\;=\;
\big[\underline R(o),\ \overline R(\min\{1,\tfrac{K}{n}o\})\big],
\]
which is a strictly larger interval than $[\underline R(B),\overline R(B)]$ for any particular $B$ (including the true
$B=\mathrm{WPAUC}(\alpha,d)$), because \eqref{eq:B-interval-from-o} has positive length and $\underline R,\overline R$
are non-constant on $(0,1)$.
Therefore, for the same $K$, the interval induced by \eqref{eq:wpauc-recall-bound} (conditioning on $B$) is strictly tighter
than any interval obtainable from $\mathrm{OPAUC}(\alpha{+}d)$ alone (which conditions only on $o$ and must account for the
uncertainty in $B$).

\end{proof}

\subsection{Proof of Lemma~\ref{lem:wpauc-eq-recall}}
\begin{proof}
Let the unique positive item be $i^+$.
By definition, $\mathrm{Recall@K}= \mathbb{I}\{\mathrm{rank}(i^+)\le K\}$.

The windowed negative set $\mathcal{W}_u^-(\alpha,d)$ in \eqref{eq:window-neg}
contains exactly one negative item:
the $(K\!-\!1)$-th highest-scored negative, denoted $j_{(K-1)}$.
Therefore, by \eqref{eq:wpauc},
\[
\mathrm{WPAUC}_u(\alpha,d)
=
\mathbb{I}\{ f_{u,i^+} > f_{u,j_{(K-1)}}\}.
\]
The inequality $f_{u,i^+} > f_{u,j_{(K-1)}}$ holds iff the positive item is ranked above the $(K\!-\!1)$-th negative,
which is equivalent to $i^+$ appearing within the overall top-$K$ positions.
Hence $\mathrm{WPAUC}_u(\alpha,d)=\mathbb{I}\{\mathrm{rank}(i^+)\le K\}=\mathrm{Recall@K}$.
\end{proof}

\begin{algorithm}[t]
\caption{TAWin: Threshold-Adjusted Soft Window Reweighting}
\label{alg:tawin}
\begin{algorithmic}[1]
\STATE \textbf{Input:} flattened token-level advantages $\bm{a}\in\mathbb{R}^{N}$, group size $G$,
Top-$K$ mass $K$, temperature $\tau>0$, anchor $\tilde \sigma_\star\in[0,1]$.
\STATE \textbf{Output:} reweighted advantages $\bm{a}'\in\mathbb{R}^{N}$.
\STATE Set $\bm{a}'\leftarrow \bm{a}$ and $M \leftarrow N/G$ \hfill\COMMENT{$M$ groups}
\FOR{$g=1$ \textbf{to} $M$}
    \STATE $\mathcal{I}\leftarrow\{(g-1)G+1,\ldots,gG\}$; \quad $\bm{a}_g \leftarrow \bm{a}'[\mathcal{I}]$
    \STATE $\mathcal{N}\leftarrow\{i\in\mathcal{I}: a_{g,i}<0\}$; \quad $\mathcal{P}\leftarrow\{i\in\mathcal{I}: a_{g,i}\ge 0\}$
    \IF{$|\mathcal{N}|=0$} \STATE \textbf{continue} \ENDIF
    \STATE $K \leftarrow |\mathcal{N}|$
    \STATE \textbf{Rank negatives:} sort $\mathcal{N}$ by model score $f$ in descending order and assign ranks $\sigma\in\{1,\ldots,K\}$ \hfill\COMMENT{higher score $\Rightarrow$ smaller rank}
    \STATE \textbf{Normalize ranks:} $\tilde \sigma(\sigma)\leftarrow \frac{\sigma-1}{K-1}$; if $K=1$, define $\tilde \sigma(\sigma)\leftarrow 0$
    \STATE \textbf{TAWin logits:} for each negative with rank $r$, set
    \[
        x(\sigma)\leftarrow -\big|\tilde \sigma(\sigma)-\tilde \sigma(\sigma_\star)\big|
    \]
    \STATE \textbf{Sort logits:} let $x_{(1)}\ge \cdots \ge x_{(K)}$ denote logits in non-increasing order
    \STATE \textbf{Compute threshold $\lambda$:} find an index $m\in\{0,\ldots,K-1\}$ such that $x_{(m)}\ge \lambda \ge x_{(m+1)}$ (with $x_{(0)}=+\infty$), and set
    \[
        \lambda \leftarrow \tau\log\!\Big(\sum_{i=m+1}^{K}\exp\!\big(\tfrac{x_{(i)}}{\tau}\big)\Big)-\tau\log(K-m)
    \]
    \STATE \textbf{ThreTopK weights (sorted order):} for $i=1,\ldots,K$,
    \[
        w_{(i)} \leftarrow \min\!\left\{1,\exp\!\big(\tfrac{x_{(i)}-\lambda}{\tau}\big)\right\}
    \]
    \STATE \textbf{Restore original order:} map $\{w_{(i)}\}$ back to indices $i\in\mathcal{N}$ to obtain $\{w_i\}_{i\in\mathcal{N}}$
    \STATE \textbf{Mass normalization:} for all $i\in\mathcal{N}$, set
    \[
        \bar w_i \leftarrow \frac{w_i}{\sum_{j\in\mathcal{N}} w_j}\cdot |\mathcal{N}|
    \]
    \STATE \textbf{Reweight negative advantages:} $a'_{g,i}\leftarrow a'_{g,i}\cdot \bar w_i \quad \forall i\in\mathcal{N}$
\ENDFOR
\STATE \textbf{return} $\bm{a}'$
\end{algorithmic}
\end{algorithm}

\section{TAWin Algorithm}\label{appendix:algorigthm}

The detailed algorithm of TAWin is provided in~\ref{alg:tawin}.

\section{Experimental Settings}

\subsection{Dataset}\label{appendix:dataset}
We construct our datasets from the Amazon Review corpus following the preprocessing pipeline in~\citet{bao2024decoding, DBLP:journals/corr/abs-2510-12211}. First, we restrict interactions to a fixed temporal window defined per category: from October 2016 to November 2018 for Toys and Games and Office Products, and from October 1996 to November 2018 for Industrial and Scientific. For Yelp, we select the interactions from the year of 2021. We then apply iterative K-core filtering with $K=5$, removing users and items with fewer than five interactions until convergence. To ensure high-quality textual representations, we further discard items with missing or noisy metadata and retain only those with valid and concise titles.

For each remaining user, interactions are sorted chronologically and transformed into supervised samples using a sliding window, where the historical context is truncated to at most ten preceding items and the subsequent item is treated as the prediction target. This procedure yields multiple training instances per user while preserving temporal order. Finally, all samples are globally sorted by timestamp and split chronologically into training, validation, and test sets using an 8:1:1 ratio. Dataset statistics for the resulting training splits are reported in Table~\ref{tab:dataset_statistics}.

\begin{table}[t]
\centering
\begin{tabular}{lcccc}
\toprule
\textbf{Datasets} & \textbf{Toys} & \textbf{Industrial} & \textbf{Office} & \textbf{Yelp} \\
\midrule
Items & 11,252 & 3,685 & 3,459 & 8,785 \\
Train & 112,754 & 36,259 & 38,924 & 77,097 \\
Valid & 14,095 & 4,532 & 4,866 & 9,637 \\
Test  & 14,095 & 4,533 & 4,866 & 9,638 \\
\bottomrule
\end{tabular}
\caption{Statistics of datasets.}
\label{tab:dataset_statistics}
\end{table}

\subsection{Evaluation Metrics.}\label{appendix:metrics}
To assess Top-$K$ recommendation performance, we adopt \textbf{Recall (R@K)} and
\textbf{Normalized Discounted Cumulative Gain (N@K)}.
For each user $u$, we apply the constrained beam search generation strategy $\mathcal{G}$ to obtain a ranked list of
$K$ recommended items:
\begin{equation}
\mathcal{R}_u^K
~\triangleq~
\bigl\{ \phi(Y_1), \ldots, \phi(Y_K) \bigr\},
\end{equation}
where $\{Y_1,\ldots,Y_K\}$ are the Top-$K$ completed sequences returned by
$\mathcal{G}(\pi_\theta, H_u, \mathcal{I})$, ordered by model likelihood.

\paragraph{Recall (R@K).}
R@K measures whether the ground-truth target item $i_t$ appears in the Top-$K$ recommendation list:
\begin{equation}
\mathrm{R@}K
~\triangleq~
\frac{1}{|\mathcal{U}|}
\sum_{u\in\mathcal{U}}
\mathbb{I}\!\left[\, i_t \in \mathcal{R}_u^K \,\right],
\end{equation}
where $\mathbb{I}[\cdot]$ denotes the indicator function.

\paragraph{Normalized Discounted Cumulative Gain (N@K).}
N@K further accounts for the rank position of the target item within the Top-$K$ list.
Let $\mathrm{rank}_u(i_t)$ denote the rank of $i_t$ in $\mathcal{R}_u^K$ (smaller is better). Then
\begin{equation}
\mathrm{N@}K
~\triangleq~
\frac{1}{|\mathcal{U}|}
\sum_{u\in\mathcal{U}}
\frac{\mathbb{I}\!\left[\, i_t \in \mathcal{R}_u^K \,\right]}
{\log_2\!\bigl(\mathrm{rank}_u(i_t)+1\bigr)}.
\end{equation}

\subsection{Baselines}\label{appendix:Baselines}

We compare TAWin with representative baselines from three categories:
traditional sequential recommendation models, generative recommendation models,
and LLM-based recommendation models. The baselines are described as follows:

\begin{itemize}

  \item \textbf{GRU4Rec}~\citep{hidasi2015session} models session-based recommendation by applying recurrent neural networks to entire interaction sequences, enabling temporal dependency modeling with a ranking-oriented loss.

  \item \textbf{Caser}~\citep{tang2018personalized} embeds recent interaction sequences into a two-dimensional representation and applies horizontal and vertical convolutional filters to capture both general preferences and local sequential patterns.

  \item \textbf{SASRec}~\citep{kang2018self} employs a self-attention mechanism to adaptively select relevant items from a user’s interaction history, balancing long-term dependency modeling and efficient sequential prediction.

  \item \textbf{TIGER}~\citep{DBLP:conf/nips/RajputMSKVHHT0S23} formulates recommendation as a generative retrieval task by predicting semantically structured item identifiers using a Transformer-based sequence-to-sequence model.

  \item \textbf{LC-Rec}~\citep{zheng2024adapting} bridges language and collaborative semantics by learning quantized item indices and aligning LLMs to recommendation tasks through specialized fine-tuning objectives.

  \item \textbf{MiniOneRec}~\citep{kong2025minionerec} presents an open-source generative recommendation framework that constructs Semantic IDs and applies supervised and reinforcement learning to study scaling behaviors of LLM-based recommenders.

  \item \textbf{BigRec}~\citep{DBLP:journals/tors/BaoZWZYLCFT25} proposes a bi-step grounding paradigm that first aligns LLM outputs to recommendation tokens and then maps them to valid items to achieve full-corpus ranking capability.

  \item \textbf{D$^3$}~\citep{bao2024decoding} introduces a debiasing and diversifying decoding strategy that removes length normalization effects and incorporates auxiliary models to improve accuracy and recommendation diversity.

  \item \textbf{S-DPO}~\citep{DBLP:conf/nips/ChenTZ0SZWC24} extends Direct Preference Optimization to recommendation by incorporating multiple negatives and optimizing a ranking-aware objective derived from partial Plackett--Luce models.

  \item \textbf{ReRe}~\citep{DBLP:journals/corr/abs-2510-12211} optimizes LLM-based recommenders via reinforcement learning with verifiable rewards, combining constrained generation and auxiliary ranking supervision to improve hard-negative modeling.

\end{itemize}

\subsection{Implement Details}\label{appendix:implement_details}

Traditional recommendation models are trained with binary cross-entropy loss and the Adam optimizer.
The learning rate is tuned over $\{10^{-2},10^{-3},10^{-4}\}$ and the weight decay is selected from
$\{10^{-2},10^{-3},10^{-4},10^{-5},10^{-6}\}$; the batch size is fixed to 1024.
For TIGER, we adopt T5~\citep{DBLP:journals/corr/abs-1910-10683} as the backbone.

For LLM-based recommenders, we use Qwen2.5-0.5B~\citep{qwen2025qwen25technicalreport} as the base model to reduce computational overhead,
and optimize all models with AdamW~\citep{DBLP:journals/corr/abs-1711-05101}.
During fine-tuning, SFT and preference-alignment data are processed with batch size 128, while reinforcement learning
uses batch size 512.
We apply a learning rate of $3\times 10^{-4}$ for SFT and $1\times 10^{-5}$ for S-DPO, SPRec, and ReRe,
all with a cosine learning-rate scheduler.
SFT runs for up to 10 epochs with early stopping (patience 1).
DPO-based models (S-DPO, SPRec) are trained for 1 epoch, while ReRe is trained for 2 epochs.

For D$^{3}$, the interpolation coefficient $\alpha$ is searched over $\{0.8,0.9,1.0\}$.
For DPO-based methods, we set the KL coefficient to $\beta=0.1$, and S-DPO samples 3 negative items per instance.
For ReRe, we set $\beta=10^{-3}$ and generate 16 candidate items per prompt.
When implementing GRPO, we estimate the KL divergence via a second-order approximation.

\section{Reweighting Ablation}\label{appendix:selection_ablaton}

\begin{table*}[t]
\centering
\caption{Performance comparison between ReRe with ranking reward ($\text{ReRe}^*$) and TAWin. The best performance is
highlighted in boldface.}
\label{tab:selection_ablation}
\resizebox{\textwidth}{!}{
\begin{tabular}{l ccc ccc ccc ccc}
\toprule
\multirow{2}{*}{Method} & 
\multicolumn{3}{c}{Toys} & 
\multicolumn{3}{c}{Industrial} & 
\multicolumn{3}{c}{Office} & 
\multicolumn{3}{c}{Yelp} \\
\cmidrule(lr){2-4}
\cmidrule(lr){5-7}
\cmidrule(lr){8-10}
\cmidrule(lr){11-13}

& R@1 & R@3 & N@3
& R@1 & R@3 & N@3
& R@1 & R@3 & N@3
& R@1 & R@3 & N@3 \\
\midrule

$\text{ReRe}^*$
& 0.0442 & 0.0740 & 0.0615
& 0.0882 & 0.1224 & 0.1091
& 0.0931 & 0.1333 & 0.1170
&  0.0206 & 0.0358 & 0.0295
 \\

TAWin
& \textbf{0.0471} & \textbf{0.0761} & \textbf{0.0639} 
& \textbf{0.0904} & \textbf{0.1237} & \textbf{0.1099}
& \textbf{0.0961} & \textbf{0.1341} & \textbf{0.1187}
&  \textbf{0.0227} & \textbf{0.0370} & \textbf{0.0301}
 \\

\bottomrule
\end{tabular}
}

\end{table*}

\paragraph{ReRe with ranking reward as a special case.}
In addition to the rule-based reward used by ReRe, we consider its \emph{ranking-reward} variant and denote it as $\text{ReRe}^*$.
We show that $\text{ReRe}^*$ can be interpreted as a \emph{special reweighting scheme} over beam-search negatives: rather than changing the underlying GRPO-style update, the ranking reward induces deterministic, rank-dependent weights on sampled candidates, thereby scaling each candidate’s policy-gradient contribution according to its relative position within the beam set.
Formally, the induced objective of $\text{ReRe}^*$ is equivalent to optimizing a reweighted surrogate where the weight is a fixed function of the candidate rank (or normalized rank) under the current policy.
The precise equivalence, including the explicit mapping from ranking rewards to the corresponding weight function, is provided in Proposition~\ref{prop:rank-reweight-icml}.

\begin{proposition}[Ranking Reward as Sequence-Level Reweighting (One Positive per Group)]
\label{prop:rank-reweight-icml}
Fix a prompt $H_u$ and draw a group of $G$ rollouts $\{Y_k\}_{k=1}^{G}$ with item titles $e_k=\phi(Y_k)$.
Assume the group contains \emph{exactly one} positive title: there exists a unique index $t$ such that $e_t=e_\star$ (the target title), and $e_k\neq e_\star$ for all $k\neq t$.
Let the rule-based reward be $R^{(0)}_k \triangleq \mathbb{I}[e_k=e_\star]$, and define the ranking-shaped reward
\begin{equation}
R_k \;\triangleq\; R^{(0)}_k + R_{\rm rank}(e_k,e_\star),
\qquad\text{with}\qquad
R_{\rm rank}(e_\star,e_\star)=0.
\end{equation}
Let $\bar R \triangleq \frac{1}{G}\sum_{i=1}^{G} R_i$ and $\hat A_k \triangleq R_k-\bar R$ be the group-mean baseline and advantages.
Then there exist sequence-level weights $\{\omega_k\}_{k=1}^{G}$ such that
\begin{equation}
\hat A_k \;=\; \omega_k\,\hat A^{(0)}_k,\qquad \forall k,
\end{equation}
where $\hat A^{(0)}_k \triangleq R^{(0)}_k - \frac{1}{G}$ are the advantages under the rule-only reward, and
\begin{equation}
\omega_t \;=\; \frac{1-\bar R}{1-\frac{1}{G}},
\qquad
\omega_k \;=\; \frac{R_{\rm rank}(e_k,e_\star)-\bar R}{-\frac{1}{G}},\;\; k\neq t.
\end{equation}
Consequently, in GRPO-style clipped objectives (ignoring the weak KL term), adding $R_{\rm rank}$ is exactly a
sequence-level reweighting of the rule-only objective.
\end{proposition}

\begin{proof}
Under the one-positive assumption, the rule-only rewards satisfy $R^{(0)}_t=1$ and $R^{(0)}_{k}=0$ for $k\neq t$.
Hence the group mean is $\bar R^{(0)}=\frac{1}{G}$ and the corresponding advantages are
\begin{equation}
\hat A^{(0)}_t \;=\; 1-\frac{1}{G},
\qquad
\hat A^{(0)}_{k} \;=\; -\frac{1}{G}\;\; (k\neq t).
\label{eq:ruleonly-adv}
\end{equation}
For the shaped reward $R_k=R^{(0)}_k+R_{\rm rank}(e_k,e_\star)$ with $R_{\rm rank}(e_\star,e_\star)=0$, we have
$R_t=1$ and $R_k=R_{\rm rank}(e_k,e_\star)$ for $k\neq t$, so the group advantages are
\begin{equation}
\hat A_t \;=\; 1-\bar R,
\qquad
\hat A_k \;=\; R_{\rm rank}(e_k,e_\star)-\bar R\;\; (k\neq t).
\label{eq:shaped-adv}
\end{equation}
Define $\omega_t \triangleq \hat A_t/\hat A^{(0)}_t$ and $\omega_k \triangleq \hat A_k/\hat A^{(0)}_k$ for $k\neq t$.
Since \eqref{eq:ruleonly-adv} ensures $\hat A^{(0)}_k\neq 0$ for all $k$, we obtain
\begin{equation}
\hat A_k=\omega_k \hat A^{(0)}_k,\;\;\forall k,
\end{equation}
with the explicit forms in the statement.
Finally, GRPO clipped surrogates are linear in the (sequence-level) advantage: each rollout $Y_k$ contributes a
token-average term multiplied by $\hat A_k$. Replacing $\hat A_k$ with $\omega_k \hat A^{(0)}_k$ is therefore
exactly a sequence-level reweighting of the rule-only objective, concluding the proof.
\end{proof}

\paragraph{$\text{ReRe}^*$ vs. TAWin.}
While $\text{ReRe}^*$ offers an intuitive heuristic that favors hard negatives identified by beam search, its weighting rule is not derived from an explicit Top-$K$ alignment objective.
In contrast, TAWin is directly motivated by our theoretical insight that Top-$K$ optimization can be better approached via windowed partial-AUC shaping, and it implements a differentiable, temperature-controlled reweighting that targets the desired low-FPR region.
As reported in Table~\ref{tab:selection_ablation}, TAWin consistently outperforms $\text{ReRe}^*$ across four datasets (Toys, Industrial, Office, and Yelp) in terms of both Recall and NDCG at small cutoffs, demonstrating that TAWin yields more effective preference optimization than the rank-based heuristic used in $\text{ReRe}^*$.

\end{document}